\documentclass[twocolumn,floatfix,superscriptaddress,amsmath,amssymb,aps,prl]{revtex4-2}
\pdfoutput=1
\usepackage{graphicx}
\graphicspath{ {Figures/} }

\usepackage{dcolumn}
\usepackage{bm}
\usepackage{bbold}
\usepackage[export]{adjustbox}

\usepackage{multirow}

\usepackage{mathtools} 
\DeclarePairedDelimiter\bra{\langle}{\rvert}
\DeclarePairedDelimiter\ket{\lvert}{\rangle}
\DeclarePairedDelimiterX\braket[2]{\langle}{\rangle}{#1 \delimsize\vert #2}
\DeclareMathOperator{\tr}{tr}

\usepackage{array,colortbl,xcolor}

\usepackage[utf8]{inputenc}
\usepackage[T1]{fontenc}
\usepackage{mathrsfs,bbm}
\usepackage{dsfont}

\usepackage{amssymb}
\usepackage{amsmath}
\usepackage{amsfonts}
\usepackage{amsthm}
\usepackage{mathtools}
\usepackage{hyperref}
\hypersetup{pdfpagemode=UseNone}
\hypersetup{colorlinks=true}
\hypersetup{citecolor=blue}
\hypersetup{linkcolor=blue}
\hypersetup{urlcolor=blue}
\usepackage{cleveref}
\newtheorem{theorem}{Theorem}

\newtheorem{proposition}[theorem]{Proposition}

\newcounter{rem}
\def\>{\rangle}
\def\<{\langle}

\def\tr{{\rm tr}}

\def\textbf#1{{\bf #1}}

\newcommand{\Nl}{\mathbb{N}}

\usepackage{dsfont}
\DeclareMathOperator{\id}{\mathds{1}}
\DeclareMathOperator{\range}{range}
\DeclareMathOperator{\linspan}{span}
\renewcommand{\paragraph}[1]{\addcontentsline{toc}{section}{#1}\emph{#1.}---}
\usepackage{xcolor}

\usepackage{soul}

\begin{document}

\title{Semi-device-independent certification of number of measurements}

\author{Isadora Veeren}
\affiliation{Centro Brasileiro de Pesquisas Físicas, Rua Dr. Xavier Sigaud, 150, Rio de Janeiro, RJ, Brasil}

\author{Martin Pl\'{a}vala}
\affiliation{Naturwissenschaftlich-Technische Fakult\"{a}t, Universit\"{a}t Siegen, Walter-Flex-Stra\ss e 3, 57068 Siegen, Germany}

\author{Leevi Lepp\"{a}j\"{a}rvi}
\affiliation{RCQI, Institute of Physics, Slovak Academy of Sciences,
D\'ubravsk\'a cesta 9, 84511 Bratislava, Slovakia}

\author{Roope Uola}
\affiliation{Department of Applied Physics, University of Geneva, 1211 Geneva, Switzerland}

\date{\today}

\begin{abstract} 
We develop a method for semi-device-independent certification of number of measurements. We achieve this by testing whether Bob's steering equivalent observables (SEO) can be simulated by $k$ measurements, which we do by testing whether they are $k$-compatible with separable joint observable. This test can be performed with the aid of hierarchy of semidefinite programs, and whenever it fails one can conclude that Alice must have access to at least $k+1$ incompatible measurements.
\end{abstract}

\maketitle

\paragraph{Introduction}
The theory of quantum measurements is at the heart of our understanding of quantum theory. Originally developed to describe and understand the measurement process on the foundational level, recently the advantages of this theory for quantum information processing have become clear, see \cite{heinosaari2016invitation,JMreviewOtfried} for reviews on the topic.

So far, the measurement theoretic approach to quantum correlations and communication has concentrated mainly on measurements acting on a single system. This is despite the fact that quantum networks have gone through a rapid development in recent years \cite{networkreview}, and the investigation of such platforms necessitates the use of multipartite measurements. On top of the fundamental results, such as the refutation of real-number-based models of quantum theory \cite{Renouimaginary}, the approach has resulted in an improved understanding of measurements acting on compound systems, for example, in the form of the elegant joint measurement \cite{GisinEJM}. However, the full theory of quantum measurements on such systems is still developing, and the exact role of known concepts, such as entangled measurements and measurements simulable by local operations supported by classical communication \cite{JakubISOentangled,LOCCreview}, in networks remains unclear.

In this manuscript, we contribute to the theory of measurements in compound systems by exploiting a link between two central concepts, resulting in a hierarchy for certifying the number of measurements in a semi-device independent manner. These concepts are simulability of measurements and compatibility of measurements on many copies. On the one hand, simulability asks whether the statistics of a given set of quantum measurements can be explained by using a pre-fixed set of measurements assisted by, for example, randomness and classical data processing. Recently, the concept has drawn a considerable amount of interest especially for the cases when the pre-fixed measurements are sharp \cite{Oszmaniecprojsim}, contributing to improving known bounds on the Grothendieck constant $K_G(3)$ \cite{Hirschconstant}, have fixed number of outcomes \cite{guerini2017,Shi2020}, showing that there are truly non-projective measurements in quantum theory \cite{Gomeznonprojective}, have a limited dimensionality \cite{Ioannousimulability}, allowing for a record high semi-device independent certification of entanglement dimensionality \cite{Designolle21HDSteer}, and the number of simulating measurements is fixed \cite{guerini2017}. Here we concentrate on the fourth case. On the other hand, compatibility on many copies asks whether one can recover the statistics of a set of measurements from a single measurement on a compound system \cite{carmeli2016}. This forms a natural generalization of compatibility of measurements, which is recovered when compound system is not used. So far, this special case has found deep connections with an advantage in quantum correlations in bipartite \cite{Wolf09,quintino2014,Uola14}, prepare-and-measure \cite{Tavakoli20,Saha23}, and temporal \cite{Carmeli19,Skrzypczyk19,Uola19quantifying,Uola19leggettgrag,Oszmaniecoperational,Uola20,Buscemi20,Uola22} scenarios, but the more recent and more general many-copy case has not yet found applications in such setting.

We demonstrate the use of the proposed hierarchy in a quantum steering scenario. Steering manifests itself in asymmetric scenarios where Alice and Bob cannot be inter-exchanged \cite{Midgley10,Olsen2013,Bowles2014,Evans2014,Skrzypczyk2014,Bowles2016,Sekatski2023}. Because of the way this scenario is designed, steering can be formulated in one-way-device-independent approaches. As a result, it can also be exploited for many practical applications, such as quantum key distribution \cite{branciard2012}, randomness certification \cite{law2014,Passaro2015,Skrzypczyk2018} or secret sharing \cite{xiang2017,Kogias17}, where such asymmetry can be an advantage.

Beyond its applications, steering has a considerable role in foundations of quantum mechanics. Though a completely distinct and independent phenomenon, it is closely related to Bell non-locality, entanglement and incompatibility of measurements \cite{Cavalcanti2016,steering-review}. In fact, steering can only be observed if the shared state is entangled and the measurements performed by Alice are incompatible \cite{quintino2014,Uola14}. Such requirement makes the detection of steering a strategy to certify that these measurements are incompatible in a one-way-device-independent approach \cite{Chen2016,Cavalcanti16quantitative}, a feat that is relevant on its own since incompatibility is central in many protocols.

Using our results on measurements, we go beyond witnessing incompatibility of measurements through the violation of a steering inequality. Indeed, we provide a strategy to actually certify the number of incompatible measurements to which Alice has access. In a bipartite steering scenario, we develop a series of tests that Bob can perform with the information he has on his side of the experiment and establish a lower bound to how many incompatible measurements his colleague must have performed on her part of a shared quantum state. Our method also displays the main advantage of admitting a semi-definite programming (SDP) formulation, and our tests can, as a consequence, be efficiently computed with well-established methods.

\paragraph{Compatibility and simulability}
Joint measurability of a set of positive-operator valued measures (POVMs) encapsulates the notion of whether a set of measurements can be performed simultaneously by a single device. Formally, a measurement assemblage, i.e., a set of POVMs $\{M_{a|x}\}_{a,x}$ obeying $M_{a|x} \geq 0 $ and $\sum_a M_{a|x} = \id$ for all $a,x$, is said to be compatible or jointly measurable if there exist a POVM $\{G_{\lambda}\}_{\lambda}$ and a probability distribution $p$ such that, for any state $\rho$,
\begin{equation}
    \tr[\rho M_{a|x}] =\sum_{\lambda} p(a|x,\lambda) \tr[\rho G_{\lambda}].
\end{equation}
When this condition is met, it means that the joint observable $\{G_\lambda\}_\lambda$ represents a simultaneous measurement of $\{M_{a|x}\}_{a,x}$ and can recover its statistics through a suitable post-processing. Equivalently, this condition can be put in terms of the marginalization of some joint measurement. In this case, there must be a POVM $G$ with effects $G_{a_1, \dots, a_n}$ such that $M_{a|x} = \sum_{a_i, i \neq x} G_{a_1, \dots, a_{x-1}, a, a_{x+1}, \dots ,a_n}$. Then it is said that $M_{a|x}$ can be recovered as the $x$\textsuperscript{th} marginal of $G$.

The notion of joint measurability can be extended to the concept of $k$-compatibility \cite{carmeli2016}, that is central for our work. The idea is that we are allowed to perform the joint measurement on $k$ copies of the state $\rho$. Thus a set of measurements $\{M_{a|x}\}_{a,x}$ is said to be $k$-compatible if there exists a joint observable $\{G_{\lambda}\}_{\lambda}$ such that for any state $\rho$ we have
\begin{equation}
    \tr[\rho M_{a|x}] = \sum_{\lambda} p(a|x,\lambda) \tr[\rho^{\otimes k} G_{\lambda}].
\end{equation}
In such case $G$ is called a $k$-copy joint observable of $\{M_{a|x}\}_{a,x}$. This is a relaxation of the usual concept of compatibility, which is recovered in the case where one has access to only one copy of $\rho$, meaning $k=1$. Also, notice that any set of $k$ measurements is $k$-compatible. An example of incompatible but $2$-compatible POVMs is given by the hollow triangle POVMs \cite{Heinosaari2008}, i.e. a set of 3 incompatible POVMs that are compatible pairwise. Such triplet of POVMs is clearly $2$-compatible, as one can construct the $2$-copy joint POVM by measuring one of the POVMs on one copy of $\rho$ and measuring the joint observable of the other two on the other copy of $\rho$.

Another central concept for our work is that of $k$-simulability \cite{guerini2017}. In contrast to having many copies of the state, in $k$-simulability one has $k$ measurements. Also, in contrast to joint measurability that uses only classical post-processing, in $k$-simulability one is allowed to use classical pre-processing as well. 
Formally, a set of measurements given by POVM elements $\{M_{a|x}\}_{a,x}$ is $k$-simulable if there exist probability distributions $p$ and $q$, as well as $k$ POVMs with elements $B_{b|y}$ such that
\begin{equation} \label{eq:ksim}
    M_{a|x} = \sum_{y=1}^{k} p(y|x) \sum_{b} q(a|x,b,y)B_{b|y}.
\end{equation}
When that is the case, $\{B_{b|y}\}_{b,y}$ are called $k$-simulators of $\{M_{a|x}\}_{a|x}$. Notice that, similarly to $k$-compatibility, any set of $k$ POVMs is $k$-simulable and we recover the notion of usual compatibility when $k=1$. As an example of $2$-simulable but incompatible measurements one can again take the hollow triangle POVMs and apply a construction analogous to the previous one.

\paragraph{Steering and steering equivalent observables}
Consider a scenario where two parties, Alice and Bob, share a quantum state $\rho_{AB}$, upon which they can perform measurements and classically exchange its results. On Alice's side, she has access to a set of measurements labeled by $x$ with outcomes $a$, described by the POVM effects $\{M_{a|x}\}_{a,x}$. After Alice measures, Bob is left with reduced states conditioned to Alice's measurement choice and measurement outcome. That is, Bob will have access to a state assemblage $\{\rho_{a|x}\}_{a,x}$ given as $\rho_{a|x} = \tr_{A}[(M_{a|x} \otimes \id)\rho_{AB}]$. 
Upon determining his state assemblage, Bob can check whether the assemblage can be explained by a local hidden state (LHS) model. An assemblage is said to allow LHS model if there are probabilities $p(a|x,\lambda)$ and sub-normalized states $\sigma_{\lambda}$ such that $\rho_{a|x} = \sum_{\lambda} p(a|x,\lambda)\sigma_{\lambda}$. In this case, Bob could just claim the states he observes come from local states $\sigma_{\lambda}$ in his laboratory whose probability distributions are simply updated by finding out Alice's measurement outcomes. When that is not the case,
then Bob is concludes that Alice is able to steer his states, i.e. the state assemblage can not be realized using a separable state.

Steerability of any state assemblage can be put in terms of a joint measurability problem. Let $\rho_B = \tr_A[\rho_{AB}]$ be Bob's reduced state and $\Pi_B$ be the projection onto $\range(\rho_B)$. The steering equivalent observables (SEO) of Bob's state assemblage $\{\rho_{a|x}\}_{a,x}$ are defined as $S_{a|x} = \tilde{\rho}_B^{-\frac{1}{2}} \tilde{\rho}_{a|x} \tilde{\rho}_B^{-\frac{1}{2}}$, where $\tilde{\rho}_{a|x} = \Pi_B \rho_{a|x} \Pi_B^{\dagger} $ and $\tilde{\rho}_B = \Pi_B \rho_B \Pi_B^{\dagger} $. It is known that $\{\rho_{a|x}\}_{a,x}$ has LHS model if and only if $\{S_{a|x}\}_{a|x}$ is jointly measurable \cite{uola2015one}.

\paragraph{Constructing the test}
We will now introduce a string of implications that will lead to a test that certifies that Alice must have access to more than $k$ incompatible measurements. We start by showing that if a measurement assemblage is $k$-simulable, then the SEO of the corresponding state assemblage is also $k$-simulable. Consider a measurement assemblage $\{M_{a|x}\}_{a,x}$. If the assemblage is $k$-simulable, then there exist probability distributions $p$ and $q$, as well as $k$ POVMs with effects $\{B_{b|y}\}_{b,y}$ that satisfy Eq.~\eqref{eq:ksim}. Bob's SEO are given by
\begin{equation}
    S_{a|x} = \tilde{\rho}_B^{-\frac{1}{2}} \Pi_B \tr_{A}[(M_{a|x} \otimes \id) \rho_{AB}] \Pi_B^{\dagger} \tilde{\rho}_B^{-\frac{1}{2}}. 
\end{equation}
Using Eq. \eqref{eq:ksim} we have that
\begin{equation}
\begin{split}
    S_{a|x} = &\sum_{y=1}^k \sum_{b} p(y|x) q(a|b,x,y) \\
    &\tilde{\rho}_B^{-\frac{1}{2}} \Pi_B  \tr_{A}[(B_{b|y} \otimes \id) \rho_{AB}] \Pi_B^{\dagger} \tilde{\rho}_B^{-\frac{1}{2}}.
\end{split}
\end{equation}
We can then simply identify $\tilde{B}_{b|y}$, the simulators of $S_{a|x}$, as  
\begin{align}\label{eq:decompose}
    \tilde{B}_{b|y} &= \tilde{\rho}_B^{-\frac{1}{2}} \Pi_B \tr_{A}[(B_{b|y} \otimes \id) \rho_{AB}] \Pi_B^{\dagger} \tilde{\rho}_B^{-\frac{1}{2}}.
\end{align}
$p,q$ are assumed to be probability distributions, so as long as $\{\tilde{B}_{b|y}\}_{b,y}$ is a measurement assemblage of $k$ POVMs then, by definition, $S_{a|x}$ is $k$-simulable. All that is left is to show that $\tilde{B}_{b|y} \geq 0 $ and $\sum_a \tilde{B}_{b|y} = \id$ for all $b,y$.

It is easy to see that the first condition holds since $B_{b|y} \geq 0$, we thus only need to check that $\sum_b \tilde{B}_{b|y} = \id$. We have
\begin{equation}
    \sum_b \tilde{B}_{b|y}
    = \tilde{\rho}_B^{-\frac{1}{2}} \Pi_B \rho_B \Pi_B^{\dagger} \tilde{\rho}_B^{-\frac{1}{2}} 
    = \tilde{\rho}_B^{-\frac{1}{2}}  \tilde{\rho}_B  \tilde{\rho}_B^{-\frac{1}{2}}
    = \id.
\end{equation}
With that we conclude that if a measurement assemblage is $k$-simulable, then the steering equivalent observables of the state assemblage it generates will also be $k$-simulable. Notice however that the converse is not true: if the SEO of a state assemblage is $k$-simulable it does not mean that the measurement assemblage that generated it is also $k$-simulable. As a counter example, whenever Alice and Bob share a separable state, Bob's steering equivalent observables will be $k$-simulable for any $k \geq 1$, regardless of whether Alice's measurement assemblage is $k$-simulable or not.

We now recall the result of \cite{Filippov2018} where it was shown (in the framework of general probabilistic theories) that $k$-simulability implies $k$-compatibility of the measurement assemblage such that the $k$-copy joint measurement can be chosen to be of the product form. For completeness, we formulate the proof in the case of quantum theory. We prove it by directly constructing a joint measurement for the assemblage. Consider the set $\{ M_{a|x} \}_{a,x}$ that is $k$-simulable, meaning there must exist probability distributions $p,q$ and measurement assemblage $\{B_{b|y}\}_{b,y}$ of $k$ POVMs that satisfy \eqref{eq:ksim}. Define
\begin{equation}
N_{b|y} = \underbrace{\id \otimes \dots \otimes \overbrace{B_{b|y}}^{\text{$y$\textsuperscript{th} term}} \otimes \dots \otimes \id}_{\text{$k$ terms}},
\end{equation}
it follows that
\begin{align}
    \tr[\rho M_{a|x}] &= \sum_{y=1}^k \sum_{b} p(y|x) q(a|b,x,y) \tr[\rho^{\otimes k} N_{b|y}].
\end{align}
Now, notice that the POVM elements $N_{b|y}$ can be obtained as $y$\textsuperscript{th} marginal of $\tilde{N}_{\vec{b}} = B_{b_1|1} \otimes \dots \otimes B_{b_y|y} \otimes \dots \otimes B_{a_k|k}$, namely
\begin{equation}
    N_{b|y} = \sum_{b_{i}, i \neq y} B_{b_1|1} \otimes \dots \otimes B_{b|y} \otimes \dots \otimes B_{a_k|k}.
\end{equation}
So one can conclude that
\begin{equation}
    \tr[\rho M_{a|x}] = \sum_{y=1}^k \sum_{b_1 \ldots b_k} p(y|x) q(a|b_y,x,y) \tr[\rho^{\otimes k} \tilde{N}_{\vec{b}}].
\end{equation}
We thus conclude that the set $\{ M_{a|x} \}_{a,x}$ must be $k$-compatible with $k$-copy joint observable in a product form.

The conditions that we have derived here are necessary and, hence, they enable us to develop a hierarchy of conditions to check $k$-simulability in a semi-device-independent scenario. It is an open question whether the conditions are also sufficient, or whether one needs to add additional conditions. We present sufficient conditions for $2$-compatibility in the Appendix and we relegate the question of necessary and sufficient conditions to future research.

\paragraph{Constructing the SDP hierarchy}
Given Bob's state assemblage $\{\rho_{a|x}\}_{a,x}$, one can construct its SEO $\{S_{a|x}\}_{a,x}$. We know that if Alice's measurement assemblage $\{M_{a|x}\}_{a,x}$ is $k$-simulable then so is Bob's SEO, which means that if $\{S_{a|x}\}_{a,x}$ is not $k$-simulable then neither is $\{M_{a|x}\}_{a,x}$. This way, by checking the $k$-simulability of $\{S_{a|x}\}_{a,x}$ one can extract information about Alice's measurements: if Bob's SEO are $k$-simulable the test is inconclusive, but if $\{S_{a|x}\}_{a,x}$ is not $k$-simulable then we can conclude that Alice's measurement assemblage must consist of at least $k+1$ incompatible measurements.

The task of checking the $k$-simulability of a set of POVMs cannot in general be easily computed, but one can perform the following tests: One can test the $k$-compatibility of $\{S_{a|x}\}_{a,x}$ using an SDP, or one can test the $k$-compatibility of $\{S_{a|x}\}_{a,x}$ and enforce that the $k$-copy joint observable has positive partial transpose using an SDP, or one can test whether $\{S_{a|x}\}_{a,x}$ is $k$-compatible with $k$-copy joint observable in a separable form through an SDP hierarchy. These SDPs give a series of tests and if any of these SDPs is not feasible, then we know that Alice's assemblage must consist of at least $k+1$ incompatible measurements. If all of these SDPs are feasible, then the result of the tests is inconclusive, since the existence of a separable $k$-copy joint observable does not imply that it can be selected to be in a product form.

Even though SDPs for computing the joint measurability of a set $\{M_{a|x}\}$ can be directly written from the definition of compatibility, it takes some further investigation to be able to formulate an analogous construction to $k$-compatibility. We must recall the main result in \cite{carmeli2016}, stating that the $k$-compatibility of a set $\{M_{a|x}\}_{a,x}$ is equivalent to compatibility of the set $\{\tilde{M}^k_{a|x}\}_{a,x}$, defined as
\begin{equation} \label{eq:k-comp-to-comp}
    \tilde{M}^k_{a|x} = \dfrac{1}{k} \sum_{\ell=0}^{k-1} \id^{\ell} \otimes M_{a|x} \otimes \id^{\otimes k-\ell-1}.
\end{equation}
With this result, one can construct an SDP formulation to check the $k$-compatibility of Bob's SEO. All that is left is to also require that the $k$-copy joint observables are also separable. From Eq. \eqref{eq:k-comp-to-comp} it is clear that if the $k$-copy joint observable of $\{M_{a|x}\}$ is of the product form then the joint observable of $\{\tilde{M}_{a|x}\}$ is separable.

A bipartite operator $X \in \mathcal{L}(\mathcal{H}_A \otimes \mathcal{H}_B)$, where $\mathcal{L}(\mathcal{H}_A \otimes \mathcal{H}_B)$ denotes the set of linear operators on the tensor product of Hilbert spaces $\mathcal{H}_A$ and $\mathcal{H}_B$, is said to be separable if it can be written as $X = \sum_i Y_i \otimes Z_i$, where $Y_i, Z_i$ are positive operators on $\mathcal{H}_A$ and $\mathcal{H}_B$ respectively. Determining whether an operator can be put in this form is an NP-hard problem but there are many separability criteria that can be used tackle this problem, most famously the Positive Partial Transpose (PPT) criterion, stating that if $X$ is separable then its partial transpose must be positive, i.e. $X^{T_A} \geq 0$. Here $T_A$ denotes the partial transpose over the system $A$, defined as $(X_A \otimes X_B)^{T_A} = X_A^T \otimes X_B$.

One can consider the DPS hierarchy of criteria established in \cite{dps}, where symmetric extensions of $X \in \mathcal{L}(\mathcal{H}_A \otimes \mathcal{H}_{B_1})$ are constructed, namely operators $\tilde{X}_N \in \mathcal{L}(\mathcal{H}_A \otimes \mathcal{H}_{B_1} \otimes \ldots \otimes \mathcal{H}_{B_N})$ such that $\tr_{B_2, \ldots, B_N}[\tilde{X}_N] = X$ and $\tilde{X}_N = P \tilde{X}_N P$, where $P$ is the operator that performs any permutation of $\mathcal{H}_{B_1}, \ldots, \mathcal{H}_{B_N}$. If a certain operator $X$ has symmetric extension for arbitrary $N \in \Nl$, then $X$ is separable. This construction can be easily generalized to proving full separability of multipartite operators. In this case one needs to search for symmetric extensions over all parties but one, see \cite{aubrun2022monogamy}.

With this construction, one can build an SDP hierarchy to test whether Bob's SEO is $k$-compatible with separable $k$-copy joint observable. Since the $N$\textsuperscript{th} level of the hierarchy corresponds to symmetrically extending its POVM elements to $N$ copies, we can also apply the PPT criterion to them to improve convergence of the hierarchy. Whenever the test fails one can be sure that $\{S_{a|x}\}_{a,x}$ are not $k$-compatible with $k$-copy joint observable in a product form, and hence, after evaluating the string of implications we constructed, Alice must have access to at least $k+1$ incompatible measurements. For $N=2$ we obtain the following SDP:
\begin{align*}
 &\text{given} && \tilde{S}^k \\
 &\text{find} && \tilde{G}  \\
 &\text{s.t.} && \\
 &\text{$k$-compatibility} &&
 \begin{cases}
    & \tilde{S}^k_{a|x} = \sum_{\lambda} p(a|x,\lambda)G_{\lambda}, \forall a,x \\
    & \sum_{\lambda}G_{\lambda} = \id \\
    & G_{\lambda} \geq 0, \forall \lambda
 \end{cases} \\
 &\text{PPT} &&
 \begin{cases}
    & \tilde{G}_{\lambda}^{T_X} \geq 0, X = \{ A, B_1, B_2 \} \\
 \end{cases} \\
 &\text{first level of DPS} &&
 \begin{cases}
    & \tr_{B_2}[\tilde{G}_{\lambda}] = \tr_{B_1}[\tilde{G}_{\lambda}] = G_{\lambda} \\
    & \tilde{G}_{\lambda} \geq 0
 \end{cases}
\end{align*}
Whenever this task has no solution one can conclude that Alice has access to at least $k+1$ incompatible measurements. If there is $\tilde{G}$ that obeys these constraints the test is inconclusive and one can consider higher levels of the separability hierarchy.

\paragraph{Examples}
To test the efficiency of this test in certifying a lower bound for the number of Alice's measurements, one can investigate the typical example of measurements in mutually unbiased bases (MUB). Consider the noisy version of Alice's measurement assemblage $\{M_{a|x}\}$, consisting of $n_m$ measurements with $n_a$ outcomes each, parameterized by the visibility $t \in [0,1]$ as $M^t_{a|x} = t M_{a|x} + (1-t) \frac{\id}{n_a}$. For $t=0$ we have simply a trivial set of measurements, which is clearly $k$-compatible with separable $k$-copy joint measurement, and for $t=1$ we recover the original MUB measurements. One can evaluate what is the lowest value of the critical visibility $t_c$ for which $\{M^t_{a|x}\}_{a,x}$ passes the test, meaning for $t \geq t_c$ we can certify that Alice's assemblage consists of at least $k+1$ incompatible measurements.

For qubits we do not need to consider the DPS hierarchy, since in this case an operator is separable if and only if it is PPT. We thus obtain that for all three MUBs, the critical visibility for 2-compatibility is $\frac{\sqrt{3}}{2}$, while the critical visibility for 2-compatibility with separable $k$-copy joint observable is $\sqrt{\frac{2}{3}}$.

For qutrits our findings are summarized in Table~\ref{tab-qutrit}. The full criteria to be evaluated is that the set $\{{S}^k_{a|x}\}_{a,x}$ must be $k$-compatible with separable $k$-copy joint observable. We provide three upper estimates: $k$-compatibility, $k$-compatibility with PPT, that is, the case when we enforce that the $k$-copy joint observable is PPT, and $k$-compatiblity with PPT and first level DPS, that is, the case when we enforce that the $k$-copy joint observable is PPT and satisfies the first level of the DPS hierarchy.

\begin{table}
\begin{ruledtabular}
\begin{tabular}{ccccc}
$n_m$  & $k$ & $k$-compatibility  & PPT & first level DPS  \\
3 & 2 & 0.8553 & 0.7975 & 0.7393 \\
4 & 2 & 0.7681 & 0.6959 & 0.6933 \\
4 & 3 & 0.9310 & 0.8959 & 0.8592 \\
\end{tabular}
\end{ruledtabular}
\caption{Critical values $t_c$ for the visibility below which the set $\{{S}^k_{a|x}\}$ is $k$-compatible, $k$-compatibile with PPT, i.e., we enforce that the $k$-copy joint observable is PPT, and $k$-compatible with PPT and first level DPS, i.e., we enforce that the $k$-copy joint observable is PPT and satisfies the first level of the DPS hierarchy. $n_m$ denotes the number of measurements used to generate $\{{S}^k_{a|x}\}$, the state was taken to be the maximally entangled state.\label{tab-qutrit}}
\end{table}

\paragraph{Conclusion}
We have developed a hierarchy of conditions that enable semi-device-independent certification of number of measurements. We implemented the first levels of the hierarchy using freely accessible software. As a special case of interest, we investigated mutually unbiased bases measured on a maximally entangled state under a noisy environment, a setting that is well in the reach of today's experimental techniques. To be more specific, our results show that if Bob wants to verify that Alice has access to more than 2 measurements, this can be done in presence of significant background noise.

Our results can also be used to check $k$-simulability of POVMs. Our conditions are necessary, which is in contrast to the fact that previously there were no known conditions. There are also several possible directions for future work. The most significant is to find conditions for $k$-simulability that are both necessary and sufficient, and can be efficiently checked numerically. There are several scenarios in quantum information that can benefit from the use of notions researched in this work, we will demonstrate this by presenting two examples: One of the main obstacles in applications of device-independent quantum key distribution protocols \cite{nadlinger2022experimental,zhang2022device} is the low robustness of the protocols to various forms of noise present in the experiments. It is known that using higher dimensional Bell inequalities and Bell inequalities with more than two inputs and outcomes \cite{zapatero2023advances,miklin2022exponentially,xu2023graph,brown2021computing,gonzales2021device} can improve the noise thresholds necessary for secure quantum key distribution. The methods developed in this paper can be used by the parties in the quantum key distribution protocol to certify to each other that they have the necessary number of incompatible measurements, which are necessary to implement these protocols. Also the results hint at the possibility of semi-device-independent cryptography based on more than two measurements, which may again improve the feasibility of experimental applications. Another potential use is in attacking unknown quantum device; in every such attack there is a discovery phase where the attacker's goal is to learn as much as possible about the unknown device. The presented results show that if the attacker is able to send entangled states to the unknown device, then they can bound the number of incompatible measurements the unknown device can perform. 

\begin{acknowledgments}
\paragraph{Acknowledgments}
IV acknowledges the financial support of National Council for Scientific and Technological Development, CNPq Brazil, and thanks Carlos de Gois and Ties-A. Ohst for discussions.

MP acknowledges support from the Deutsche Forschungsgemeinschaft (DFG, German Research Foundation, project numbers 447948357 and 440958198), the Sino-German Center for Research Promotion (Project M-0294), the ERC (Consolidator Grant 683107/TempoQ), the German Ministry of Education and Research (Project QuKuK, BMBF Grant No. 16KIS1618K), and the Alexander von Humboldt Foundation.

LL acknowledges support from the European Union’s Horizon 2020 Research and Innovation Programme under the Programme SASPRO 2 COFUND Marie Sklodowska-Curie grant agreement No. 945478 as well as from projects APVV-22-0570 (DeQHOST) and VEGA 2/0183/21 (DESCOM).

RU is thankful for the financial support from the Swiss National Science Foundation (Ambizione PZ00P2-202179).
\end{acknowledgments}

\bibliography{refs}

\begin{thebibliography}{57}%
\makeatletter
\providecommand \@ifxundefined [1]{%
 \@ifx{#1\undefined}
}%
\providecommand \@ifnum [1]{%
 \ifnum #1\expandafter \@firstoftwo
 \else \expandafter \@secondoftwo
 \fi
}%
\providecommand \@ifx [1]{%
 \ifx #1\expandafter \@firstoftwo
 \else \expandafter \@secondoftwo
 \fi
}%
\providecommand \natexlab [1]{#1}%
\providecommand \enquote  [1]{``#1''}%
\providecommand \bibnamefont  [1]{#1}%
\providecommand \bibfnamefont [1]{#1}%
\providecommand \citenamefont [1]{#1}%
\providecommand \href@noop [0]{\@secondoftwo}%
\providecommand \href [0]{\begingroup \@sanitize@url \@href}%
\providecommand \@href[1]{\@@startlink{#1}\@@href}%
\providecommand \@@href[1]{\endgroup#1\@@endlink}%
\providecommand \@sanitize@url [0]{\catcode `\\12\catcode `\$12\catcode
  `\&12\catcode `\#12\catcode `\^12\catcode `\_12\catcode `\%12\relax}%
\providecommand \@@startlink[1]{}%
\providecommand \@@endlink[0]{}%
\providecommand \url  [0]{\begingroup\@sanitize@url \@url }%
\providecommand \@url [1]{\endgroup\@href {#1}{\urlprefix }}%
\providecommand \urlprefix  [0]{URL }%
\providecommand \Eprint [0]{\href }%
\providecommand \doibase [0]{https://doi.org/}%
\providecommand \selectlanguage [0]{\@gobble}%
\providecommand \bibinfo  [0]{\@secondoftwo}%
\providecommand \bibfield  [0]{\@secondoftwo}%
\providecommand \translation [1]{[#1]}%
\providecommand \BibitemOpen [0]{}%
\providecommand \bibitemStop [0]{}%
\providecommand \bibitemNoStop [0]{.\EOS\space}%
\providecommand \EOS [0]{\spacefactor3000\relax}%
\providecommand \BibitemShut  [1]{\csname bibitem#1\endcsname}%
\let\auto@bib@innerbib\@empty
\bibitem [{\citenamefont {Heinosaari}\ \emph {et~al.}(2016)\citenamefont
  {Heinosaari}, \citenamefont {Miyadera},\ and\ \citenamefont
  {Ziman}}]{heinosaari2016invitation}%
  \BibitemOpen
  \bibfield  {author} {\bibinfo {author} {\bibfnamefont {T.}~\bibnamefont
  {Heinosaari}}, \bibinfo {author} {\bibfnamefont {T.}~\bibnamefont
  {Miyadera}},\ and\ \bibinfo {author} {\bibfnamefont {M.}~\bibnamefont
  {Ziman}},\ }\bibfield  {title} {\bibinfo {title} {An invitation to quantum
  incompatibility},\ }\href {https://doi.org/10.1088/1751-8113/49/12/123001}
  {\bibfield  {journal} {\bibinfo  {journal} {Journal of Physics A:
  Mathematical and Theoretical}\ }\textbf {\bibinfo {volume} {49}},\ \bibinfo
  {pages} {123001} (\bibinfo {year} {2016})}\BibitemShut {NoStop}%
\bibitem [{\citenamefont {Gühne}\ \emph {et~al.}(2023)\citenamefont {Gühne},
  \citenamefont {Haapasalo}, \citenamefont {Kraft}, \citenamefont
  {Pellonpää},\ and\ \citenamefont {Uola}}]{JMreviewOtfried}%
  \BibitemOpen
  \bibfield  {author} {\bibinfo {author} {\bibfnamefont {O.}~\bibnamefont
  {Gühne}}, \bibinfo {author} {\bibfnamefont {E.}~\bibnamefont {Haapasalo}},
  \bibinfo {author} {\bibfnamefont {T.}~\bibnamefont {Kraft}}, \bibinfo
  {author} {\bibfnamefont {J.-P.}\ \bibnamefont {Pellonpää}},\ and\ \bibinfo
  {author} {\bibfnamefont {R.}~\bibnamefont {Uola}},\ }\bibfield  {title}
  {\bibinfo {title} {Colloquium: Incompatible measurements in quantum
  information science},\ }\href {https://doi.org/10.1103/revmodphys.95.011003}
  {\bibfield  {journal} {\bibinfo  {journal} {Reviews of Modern Physics}\
  }\textbf {\bibinfo {volume} {95}},\ \bibinfo {pages} {011003} (\bibinfo
  {year} {2023})}\BibitemShut {NoStop}%
\bibitem [{\citenamefont {Tavakoli}\ \emph {et~al.}(2022)\citenamefont
  {Tavakoli}, \citenamefont {Pozas-Kerstjens}, \citenamefont {Luo},\ and\
  \citenamefont {Renou}}]{networkreview}%
  \BibitemOpen
  \bibfield  {author} {\bibinfo {author} {\bibfnamefont {A.}~\bibnamefont
  {Tavakoli}}, \bibinfo {author} {\bibfnamefont {A.}~\bibnamefont
  {Pozas-Kerstjens}}, \bibinfo {author} {\bibfnamefont {M.-X.}\ \bibnamefont
  {Luo}},\ and\ \bibinfo {author} {\bibfnamefont {M.-O.}\ \bibnamefont
  {Renou}},\ }\bibfield  {title} {\bibinfo {title} {Bell nonlocality in
  networks},\ }\href {https://doi.org/10.1088/1361-6633/ac41bb} {\bibfield
  {journal} {\bibinfo  {journal} {Reports on Progress in Physics}\ }\textbf
  {\bibinfo {volume} {85}},\ \bibinfo {pages} {056001} (\bibinfo {year}
  {2022})}\BibitemShut {NoStop}%
\bibitem [{\citenamefont {Renou}\ \emph {et~al.}(2021)\citenamefont {Renou},
  \citenamefont {Trillo}, \citenamefont {Weilenmann}, \citenamefont {Le},
  \citenamefont {Tavakoli}, \citenamefont {Gisin}, \citenamefont
  {Ac{\'{\i}}n},\ and\ \citenamefont {Navascu{\'{e}}s}}]{Renouimaginary}%
  \BibitemOpen
  \bibfield  {author} {\bibinfo {author} {\bibfnamefont {M.-O.}\ \bibnamefont
  {Renou}}, \bibinfo {author} {\bibfnamefont {D.}~\bibnamefont {Trillo}},
  \bibinfo {author} {\bibfnamefont {M.}~\bibnamefont {Weilenmann}}, \bibinfo
  {author} {\bibfnamefont {T.~P.}\ \bibnamefont {Le}}, \bibinfo {author}
  {\bibfnamefont {A.}~\bibnamefont {Tavakoli}}, \bibinfo {author}
  {\bibfnamefont {N.}~\bibnamefont {Gisin}}, \bibinfo {author} {\bibfnamefont
  {A.}~\bibnamefont {Ac{\'{\i}}n}},\ and\ \bibinfo {author} {\bibfnamefont
  {M.}~\bibnamefont {Navascu{\'{e}}s}},\ }\bibfield  {title} {\bibinfo {title}
  {Quantum theory based on real numbers can be experimentally falsified},\
  }\href {https://doi.org/10.1038/s41586-021-04160-4} {\bibfield  {journal}
  {\bibinfo  {journal} {Nature}\ }\textbf {\bibinfo {volume} {600}},\ \bibinfo
  {pages} {625} (\bibinfo {year} {2021})}\BibitemShut {NoStop}%
\bibitem [{\citenamefont {Gisin}(2019)}]{GisinEJM}%
  \BibitemOpen
  \bibfield  {author} {\bibinfo {author} {\bibfnamefont {N.}~\bibnamefont
  {Gisin}},\ }\bibfield  {title} {\bibinfo {title} {Entanglement 25 years after
  quantum teleportation: Testing joint measurements in quantum networks},\
  }\href {https://doi.org/10.3390/e21030325} {\bibfield  {journal} {\bibinfo
  {journal} {Entropy}\ }\textbf {\bibinfo {volume} {21}},\ \bibinfo {pages}
  {325} (\bibinfo {year} {2019})}\BibitemShut {NoStop}%
\bibitem [{\citenamefont {Czartowski}\ \emph {et~al.}(2020)\citenamefont
  {Czartowski}, \citenamefont {Goyeneche}, \citenamefont {Grassl},\ and\
  \citenamefont {{\.{Z}}yczkowski}}]{JakubISOentangled}%
  \BibitemOpen
  \bibfield  {author} {\bibinfo {author} {\bibfnamefont {J.}~\bibnamefont
  {Czartowski}}, \bibinfo {author} {\bibfnamefont {D.}~\bibnamefont
  {Goyeneche}}, \bibinfo {author} {\bibfnamefont {M.}~\bibnamefont {Grassl}},\
  and\ \bibinfo {author} {\bibfnamefont {K.}~\bibnamefont {{\.{Z}}yczkowski}},\
  }\bibfield  {title} {\bibinfo {title} {Isoentangled mutually unbiased bases,
  symmetric quantum measurements, and mixed-state designs},\ }\href
  {https://doi.org/10.1103/physrevlett.124.090503} {\bibfield  {journal}
  {\bibinfo  {journal} {Physical Review Letters}\ }\textbf {\bibinfo {volume}
  {124}},\ \bibinfo {pages} {090503} (\bibinfo {year} {2020})}\BibitemShut
  {NoStop}%
\bibitem [{\citenamefont {Chitambar}\ \emph {et~al.}(2014)\citenamefont
  {Chitambar}, \citenamefont {Leung}, \citenamefont {Man{\v{c}}inska},
  \citenamefont {Ozols},\ and\ \citenamefont {Winter}}]{LOCCreview}%
  \BibitemOpen
  \bibfield  {author} {\bibinfo {author} {\bibfnamefont {E.}~\bibnamefont
  {Chitambar}}, \bibinfo {author} {\bibfnamefont {D.}~\bibnamefont {Leung}},
  \bibinfo {author} {\bibfnamefont {L.}~\bibnamefont {Man{\v{c}}inska}},
  \bibinfo {author} {\bibfnamefont {M.}~\bibnamefont {Ozols}},\ and\ \bibinfo
  {author} {\bibfnamefont {A.}~\bibnamefont {Winter}},\ }\bibfield  {title}
  {\bibinfo {title} {Everything you always wanted to know about {LOCC} (but
  were afraid to ask)},\ }\href {https://doi.org/10.1007/s00220-014-1953-9}
  {\bibfield  {journal} {\bibinfo  {journal} {Communications in Mathematical
  Physics}\ }\textbf {\bibinfo {volume} {328}},\ \bibinfo {pages} {303}
  (\bibinfo {year} {2014})}\BibitemShut {NoStop}%
\bibitem [{\citenamefont {Oszmaniec}\ \emph {et~al.}(2017)\citenamefont
  {Oszmaniec}, \citenamefont {Guerini}, \citenamefont {Wittek},\ and\
  \citenamefont {Ac{\'{\i}}n}}]{Oszmaniecprojsim}%
  \BibitemOpen
  \bibfield  {author} {\bibinfo {author} {\bibfnamefont {M.}~\bibnamefont
  {Oszmaniec}}, \bibinfo {author} {\bibfnamefont {L.}~\bibnamefont {Guerini}},
  \bibinfo {author} {\bibfnamefont {P.}~\bibnamefont {Wittek}},\ and\ \bibinfo
  {author} {\bibfnamefont {A.}~\bibnamefont {Ac{\'{\i}}n}},\ }\bibfield
  {title} {\bibinfo {title} {Simulating positive-operator-valued measures with
  projective measurements},\ }\href
  {https://doi.org/10.1103/physrevlett.119.190501} {\bibfield  {journal}
  {\bibinfo  {journal} {Physical Review Letters}\ }\textbf {\bibinfo {volume}
  {119}},\ \bibinfo {pages} {190501} (\bibinfo {year} {2017})}\BibitemShut
  {NoStop}%
\bibitem [{\citenamefont {Hirsch}\ \emph {et~al.}(2017)\citenamefont {Hirsch},
  \citenamefont {Quintino}, \citenamefont {V{\'{e}}rtesi}, \citenamefont
  {Navascu{\'{e}}s},\ and\ \citenamefont {Brunner}}]{Hirschconstant}%
  \BibitemOpen
  \bibfield  {author} {\bibinfo {author} {\bibfnamefont {F.}~\bibnamefont
  {Hirsch}}, \bibinfo {author} {\bibfnamefont {M.~T.}\ \bibnamefont
  {Quintino}}, \bibinfo {author} {\bibfnamefont {T.}~\bibnamefont
  {V{\'{e}}rtesi}}, \bibinfo {author} {\bibfnamefont {M.}~\bibnamefont
  {Navascu{\'{e}}s}},\ and\ \bibinfo {author} {\bibfnamefont {N.}~\bibnamefont
  {Brunner}},\ }\bibfield  {title} {\bibinfo {title} {Better local hidden
  variable models for two-qubit werner states and an upper bound on the
  grothendieck constant},\ }\href {https://doi.org/10.22331/q-2017-04-25-3}
  {\bibfield  {journal} {\bibinfo  {journal} {Quantum}\ }\textbf {\bibinfo
  {volume} {1}},\ \bibinfo {pages} {3} (\bibinfo {year} {2017})}\BibitemShut
  {NoStop}%
\bibitem [{\citenamefont {Guerini}\ \emph {et~al.}(2017)\citenamefont
  {Guerini}, \citenamefont {Bavaresco}, \citenamefont {Terra~Cunha},\ and\
  \citenamefont {Ac{\'\i}n}}]{guerini2017}%
  \BibitemOpen
  \bibfield  {author} {\bibinfo {author} {\bibfnamefont {L.}~\bibnamefont
  {Guerini}}, \bibinfo {author} {\bibfnamefont {J.}~\bibnamefont {Bavaresco}},
  \bibinfo {author} {\bibfnamefont {M.}~\bibnamefont {Terra~Cunha}},\ and\
  \bibinfo {author} {\bibfnamefont {A.}~\bibnamefont {Ac{\'\i}n}},\ }\bibfield
  {title} {\bibinfo {title} {Operational framework for quantum measurement
  simulability},\ }\href@noop {} {\bibfield  {journal} {\bibinfo  {journal}
  {Journal of Mathematical Physics}\ }\textbf {\bibinfo {volume} {58}},\
  \bibinfo {pages} {092102} (\bibinfo {year} {2017})}\BibitemShut {NoStop}%
\bibitem [{\citenamefont {Shi}\ and\ \citenamefont {Tang}(2020)}]{Shi2020}%
  \BibitemOpen
  \bibfield  {author} {\bibinfo {author} {\bibfnamefont {W.}~\bibnamefont
  {Shi}}\ and\ \bibinfo {author} {\bibfnamefont {C.}~\bibnamefont {Tang}},\
  }\bibfield  {title} {\bibinfo {title} {Number of quantum measurement outcomes
  as a resource},\ }\href {https://doi.org/10.1007/s11128-020-02899-9}
  {\bibfield  {journal} {\bibinfo  {journal} {Quantum Information Processing}\
  }\textbf {\bibinfo {volume} {19}},\ \bibinfo {pages} {393} (\bibinfo {year}
  {2020})}\BibitemShut {NoStop}%
\bibitem [{\citenamefont {G\'omez}\ \emph {et~al.}(2016)\citenamefont
  {G\'omez}, \citenamefont {G\'omez}, \citenamefont {Gonz\'alez}, \citenamefont
  {Ca\~nas}, \citenamefont {Barra}, \citenamefont {Delgado}, \citenamefont
  {Xavier}, \citenamefont {Cabello}, \citenamefont {Kleinmann}, \citenamefont
  {V\'ertesi},\ and\ \citenamefont {Lima}}]{Gomeznonprojective}%
  \BibitemOpen
  \bibfield  {author} {\bibinfo {author} {\bibfnamefont {E.~S.}\ \bibnamefont
  {G\'omez}}, \bibinfo {author} {\bibfnamefont {S.}~\bibnamefont {G\'omez}},
  \bibinfo {author} {\bibfnamefont {P.}~\bibnamefont {Gonz\'alez}}, \bibinfo
  {author} {\bibfnamefont {G.}~\bibnamefont {Ca\~nas}}, \bibinfo {author}
  {\bibfnamefont {J.~F.}\ \bibnamefont {Barra}}, \bibinfo {author}
  {\bibfnamefont {A.}~\bibnamefont {Delgado}}, \bibinfo {author} {\bibfnamefont
  {G.~B.}\ \bibnamefont {Xavier}}, \bibinfo {author} {\bibfnamefont
  {A.}~\bibnamefont {Cabello}}, \bibinfo {author} {\bibfnamefont
  {M.}~\bibnamefont {Kleinmann}}, \bibinfo {author} {\bibfnamefont
  {T.}~\bibnamefont {V\'ertesi}},\ and\ \bibinfo {author} {\bibfnamefont
  {G.}~\bibnamefont {Lima}},\ }\bibfield  {title} {\bibinfo {title}
  {Device-independent certification of a nonprojective qubit measurement},\
  }\href {https://doi.org/10.1103/PhysRevLett.117.260401} {\bibfield  {journal}
  {\bibinfo  {journal} {Physical Review Letters}\ }\textbf {\bibinfo {volume}
  {117}},\ \bibinfo {pages} {260401} (\bibinfo {year} {2016})}\BibitemShut
  {NoStop}%
\bibitem [{\citenamefont {Ioannou}\ \emph {et~al.}(2022)\citenamefont
  {Ioannou}, \citenamefont {Sekatski}, \citenamefont {Designolle},
  \citenamefont {Jones}, \citenamefont {Uola},\ and\ \citenamefont
  {Brunner}}]{Ioannousimulability}%
  \BibitemOpen
  \bibfield  {author} {\bibinfo {author} {\bibfnamefont {M.}~\bibnamefont
  {Ioannou}}, \bibinfo {author} {\bibfnamefont {P.}~\bibnamefont {Sekatski}},
  \bibinfo {author} {\bibfnamefont {S.}~\bibnamefont {Designolle}}, \bibinfo
  {author} {\bibfnamefont {B.~D.}\ \bibnamefont {Jones}}, \bibinfo {author}
  {\bibfnamefont {R.}~\bibnamefont {Uola}},\ and\ \bibinfo {author}
  {\bibfnamefont {N.}~\bibnamefont {Brunner}},\ }\bibfield  {title} {\bibinfo
  {title} {Simulability of high-dimensional quantum measurements},\ }\href
  {https://doi.org/10.1103/physrevlett.129.190401} {\bibfield  {journal}
  {\bibinfo  {journal} {Physical Review Letters}\ }\textbf {\bibinfo {volume}
  {129}},\ \bibinfo {pages} {190401} (\bibinfo {year} {2022})}\BibitemShut
  {NoStop}%
\bibitem [{\citenamefont {Designolle}\ \emph {et~al.}(2021)\citenamefont
  {Designolle}, \citenamefont {Srivastav}, \citenamefont {Uola}, \citenamefont
  {Valencia}, \citenamefont {McCutcheon}, \citenamefont {Malik},\ and\
  \citenamefont {Brunner}}]{Designolle21HDSteer}%
  \BibitemOpen
  \bibfield  {author} {\bibinfo {author} {\bibfnamefont {S.}~\bibnamefont
  {Designolle}}, \bibinfo {author} {\bibfnamefont {V.}~\bibnamefont
  {Srivastav}}, \bibinfo {author} {\bibfnamefont {R.}~\bibnamefont {Uola}},
  \bibinfo {author} {\bibfnamefont {N.~H.}\ \bibnamefont {Valencia}}, \bibinfo
  {author} {\bibfnamefont {W.}~\bibnamefont {McCutcheon}}, \bibinfo {author}
  {\bibfnamefont {M.}~\bibnamefont {Malik}},\ and\ \bibinfo {author}
  {\bibfnamefont {N.}~\bibnamefont {Brunner}},\ }\bibfield  {title} {\bibinfo
  {title} {Genuine high-dimensional quantum steering},\ }\href
  {https://doi.org/10.1103/physrevlett.126.200404} {\bibfield  {journal}
  {\bibinfo  {journal} {Physical Review Letters}\ }\textbf {\bibinfo {volume}
  {126}},\ \bibinfo {pages} {200404} (\bibinfo {year} {2021})}\BibitemShut
  {NoStop}%
\bibitem [{\citenamefont {Carmeli}\ \emph {et~al.}(2016)\citenamefont
  {Carmeli}, \citenamefont {Heinosaari}, \citenamefont {Reitzner},
  \citenamefont {Schultz},\ and\ \citenamefont {Toigo}}]{carmeli2016}%
  \BibitemOpen
  \bibfield  {author} {\bibinfo {author} {\bibfnamefont {C.}~\bibnamefont
  {Carmeli}}, \bibinfo {author} {\bibfnamefont {T.}~\bibnamefont {Heinosaari}},
  \bibinfo {author} {\bibfnamefont {D.}~\bibnamefont {Reitzner}}, \bibinfo
  {author} {\bibfnamefont {J.}~\bibnamefont {Schultz}},\ and\ \bibinfo {author}
  {\bibfnamefont {A.}~\bibnamefont {Toigo}},\ }\bibfield  {title} {\bibinfo
  {title} {Quantum incompatibility in collective measurements},\ }\href
  {https://doi.org/10.3390/math4030054} {\bibfield  {journal} {\bibinfo
  {journal} {Mathematics}\ }\textbf {\bibinfo {volume} {4}},\ \bibinfo {pages}
  {54} (\bibinfo {year} {2016})}\BibitemShut {NoStop}%
\bibitem [{\citenamefont {Wolf}\ \emph {et~al.}(2009)\citenamefont {Wolf},
  \citenamefont {Perez-Garcia},\ and\ \citenamefont {Fernandez}}]{Wolf09}%
  \BibitemOpen
  \bibfield  {author} {\bibinfo {author} {\bibfnamefont {M.~M.}\ \bibnamefont
  {Wolf}}, \bibinfo {author} {\bibfnamefont {D.}~\bibnamefont {Perez-Garcia}},\
  and\ \bibinfo {author} {\bibfnamefont {C.}~\bibnamefont {Fernandez}},\
  }\bibfield  {title} {\bibinfo {title} {Measurements incompatible in quantum
  theory cannot be measured jointly in any other no-signaling theory},\ }\href
  {https://doi.org/10.1103/physrevlett.103.230402} {\bibfield  {journal}
  {\bibinfo  {journal} {Physical Review Letters}\ }\textbf {\bibinfo {volume}
  {103}},\ \bibinfo {pages} {230402} (\bibinfo {year} {2009})}\BibitemShut
  {NoStop}%
\bibitem [{\citenamefont {Quintino}\ \emph {et~al.}(2014)\citenamefont
  {Quintino}, \citenamefont {V{\'e}rtesi},\ and\ \citenamefont
  {Brunner}}]{quintino2014}%
  \BibitemOpen
  \bibfield  {author} {\bibinfo {author} {\bibfnamefont {M.~T.}\ \bibnamefont
  {Quintino}}, \bibinfo {author} {\bibfnamefont {T.}~\bibnamefont
  {V{\'e}rtesi}},\ and\ \bibinfo {author} {\bibfnamefont {N.}~\bibnamefont
  {Brunner}},\ }\bibfield  {title} {\bibinfo {title} {Joint measurability,
  einstein-podolsky-rosen steering, and bell nonlocality},\ }\href
  {https://doi.org/10.1103/PhysRevLett.113.160402} {\bibfield  {journal}
  {\bibinfo  {journal} {Physical Review Letters}\ }\textbf {\bibinfo {volume}
  {113}},\ \bibinfo {pages} {160402} (\bibinfo {year} {2014})}\BibitemShut
  {NoStop}%
\bibitem [{\citenamefont {Uola}\ \emph {et~al.}(2014)\citenamefont {Uola},
  \citenamefont {Moroder},\ and\ \citenamefont {Gühne}}]{Uola14}%
  \BibitemOpen
  \bibfield  {author} {\bibinfo {author} {\bibfnamefont {R.}~\bibnamefont
  {Uola}}, \bibinfo {author} {\bibfnamefont {T.}~\bibnamefont {Moroder}},\ and\
  \bibinfo {author} {\bibfnamefont {O.}~\bibnamefont {Gühne}},\ }\bibfield
  {title} {\bibinfo {title} {Joint measurability of generalized measurements
  implies classicality},\ }\href
  {https://doi.org/10.1103/physrevlett.113.160403} {\bibfield  {journal}
  {\bibinfo  {journal} {Physical Review Letters}\ }\textbf {\bibinfo {volume}
  {113}},\ \bibinfo {pages} {160403} (\bibinfo {year} {2014})}\BibitemShut
  {NoStop}%
\bibitem [{\citenamefont {Tavakoli}\ and\ \citenamefont
  {Uola}(2020)}]{Tavakoli20}%
  \BibitemOpen
  \bibfield  {author} {\bibinfo {author} {\bibfnamefont {A.}~\bibnamefont
  {Tavakoli}}\ and\ \bibinfo {author} {\bibfnamefont {R.}~\bibnamefont
  {Uola}},\ }\bibfield  {title} {\bibinfo {title} {Measurement incompatibility
  and steering are necessary and sufficient for operational contextuality},\
  }\href {https://doi.org/10.1103/physrevresearch.2.013011} {\bibfield
  {journal} {\bibinfo  {journal} {Physical Review Research}\ }\textbf {\bibinfo
  {volume} {2}},\ \bibinfo {pages} {013011} (\bibinfo {year}
  {2020})}\BibitemShut {NoStop}%
\bibitem [{\citenamefont {Saha}\ \emph {et~al.}(2023)\citenamefont {Saha},
  \citenamefont {Das}, \citenamefont {Das}, \citenamefont {Bhattacharya},\ and\
  \citenamefont {Majumdar}}]{Saha23}%
  \BibitemOpen
  \bibfield  {author} {\bibinfo {author} {\bibfnamefont {D.}~\bibnamefont
  {Saha}}, \bibinfo {author} {\bibfnamefont {D.}~\bibnamefont {Das}}, \bibinfo
  {author} {\bibfnamefont {A.~K.}\ \bibnamefont {Das}}, \bibinfo {author}
  {\bibfnamefont {B.}~\bibnamefont {Bhattacharya}},\ and\ \bibinfo {author}
  {\bibfnamefont {A.~S.}\ \bibnamefont {Majumdar}},\ }\bibfield  {title}
  {\bibinfo {title} {Measurement incompatibility and quantum advantage in
  communication},\ }\href {https://doi.org/10.1103/physreva.107.062210}
  {\bibfield  {journal} {\bibinfo  {journal} {Physical Review A}\ }\textbf
  {\bibinfo {volume} {107}},\ \bibinfo {pages} {062210} (\bibinfo {year}
  {2023})}\BibitemShut {NoStop}%
\bibitem [{\citenamefont {Carmeli}\ \emph {et~al.}(2019)\citenamefont
  {Carmeli}, \citenamefont {Heinosaari},\ and\ \citenamefont
  {Toigo}}]{Carmeli19}%
  \BibitemOpen
  \bibfield  {author} {\bibinfo {author} {\bibfnamefont {C.}~\bibnamefont
  {Carmeli}}, \bibinfo {author} {\bibfnamefont {T.}~\bibnamefont
  {Heinosaari}},\ and\ \bibinfo {author} {\bibfnamefont {A.}~\bibnamefont
  {Toigo}},\ }\bibfield  {title} {\bibinfo {title} {Quantum incompatibility
  witnesses},\ }\href {https://doi.org/10.1103/PhysRevLett.122.130402}
  {\bibfield  {journal} {\bibinfo  {journal} {Physical Review Letters}\
  }\textbf {\bibinfo {volume} {122}},\ \bibinfo {pages} {130402} (\bibinfo
  {year} {2019})}\BibitemShut {NoStop}%
\bibitem [{\citenamefont {Skrzypczyk}\ \emph {et~al.}(2019)\citenamefont
  {Skrzypczyk}, \citenamefont {\ifmmode \check{S}\else
  \v{S}\fi{}upi\ifmmode~\acute{c}\else \'{c}\fi{}},\ and\ \citenamefont
  {Cavalcanti}}]{Skrzypczyk19}%
  \BibitemOpen
  \bibfield  {author} {\bibinfo {author} {\bibfnamefont {P.}~\bibnamefont
  {Skrzypczyk}}, \bibinfo {author} {\bibfnamefont {I.}~\bibnamefont {\ifmmode
  \check{S}\else \v{S}\fi{}upi\ifmmode~\acute{c}\else \'{c}\fi{}}},\ and\
  \bibinfo {author} {\bibfnamefont {D.}~\bibnamefont {Cavalcanti}},\ }\bibfield
   {title} {\bibinfo {title} {All sets of incompatible measurements give an
  advantage in quantum state discrimination},\ }\href
  {https://doi.org/10.1103/PhysRevLett.122.130403} {\bibfield  {journal}
  {\bibinfo  {journal} {Physical Review Letters}\ }\textbf {\bibinfo {volume}
  {122}},\ \bibinfo {pages} {130403} (\bibinfo {year} {2019})}\BibitemShut
  {NoStop}%
\bibitem [{\citenamefont {Uola}\ \emph
  {et~al.}(2019{\natexlab{a}})\citenamefont {Uola}, \citenamefont {Kraft},
  \citenamefont {Shang}, \citenamefont {Yu},\ and\ \citenamefont
  {Gühne}}]{Uola19quantifying}%
  \BibitemOpen
  \bibfield  {author} {\bibinfo {author} {\bibfnamefont {R.}~\bibnamefont
  {Uola}}, \bibinfo {author} {\bibfnamefont {T.}~\bibnamefont {Kraft}},
  \bibinfo {author} {\bibfnamefont {J.}~\bibnamefont {Shang}}, \bibinfo
  {author} {\bibfnamefont {X.-D.}\ \bibnamefont {Yu}},\ and\ \bibinfo {author}
  {\bibfnamefont {O.}~\bibnamefont {Gühne}},\ }\bibfield  {title} {\bibinfo
  {title} {Quantifying quantum resources with conic programming},\ }\bibfield
  {journal} {\bibinfo  {journal} {Physical Review Letters}\ }\textbf {\bibinfo
  {volume} {122}},\ \href {https://doi.org/10.1103/physrevlett.122.130404}
  {10.1103/physrevlett.122.130404} (\bibinfo {year}
  {2019}{\natexlab{a}})\BibitemShut {NoStop}%
\bibitem [{\citenamefont {Uola}\ \emph
  {et~al.}(2019{\natexlab{b}})\citenamefont {Uola}, \citenamefont
  {Vitagliano},\ and\ \citenamefont {Budroni}}]{Uola19leggettgrag}%
  \BibitemOpen
  \bibfield  {author} {\bibinfo {author} {\bibfnamefont {R.}~\bibnamefont
  {Uola}}, \bibinfo {author} {\bibfnamefont {G.}~\bibnamefont {Vitagliano}},\
  and\ \bibinfo {author} {\bibfnamefont {C.}~\bibnamefont {Budroni}},\
  }\bibfield  {title} {\bibinfo {title} {Leggett-garg macrorealism and the
  quantum nondisturbance conditions},\ }\href
  {https://doi.org/10.1103/physreva.100.042117} {\bibfield  {journal} {\bibinfo
   {journal} {Physical Review A}\ }\textbf {\bibinfo {volume} {100}},\ \bibinfo
  {pages} {042117} (\bibinfo {year} {2019}{\natexlab{b}})}\BibitemShut
  {NoStop}%
\bibitem [{\citenamefont {Oszmaniec}\ and\ \citenamefont
  {Biswas}(2019)}]{Oszmaniecoperational}%
  \BibitemOpen
  \bibfield  {author} {\bibinfo {author} {\bibfnamefont {M.}~\bibnamefont
  {Oszmaniec}}\ and\ \bibinfo {author} {\bibfnamefont {T.}~\bibnamefont
  {Biswas}},\ }\bibfield  {title} {\bibinfo {title} {Operational relevance of
  resource theories of quantum measurements},\ }\href
  {https://doi.org/10.22331/q-2019-04-26-133} {\bibfield  {journal} {\bibinfo
  {journal} {Quantum}\ }\textbf {\bibinfo {volume} {3}},\ \bibinfo {pages}
  {133} (\bibinfo {year} {2019})}\BibitemShut {NoStop}%
\bibitem [{\citenamefont {Uola}\ \emph
  {et~al.}(2020{\natexlab{a}})\citenamefont {Uola}, \citenamefont {Bullock},
  \citenamefont {Kraft}, \citenamefont {Pellonpää},\ and\ \citenamefont
  {Brunner}}]{Uola20}%
  \BibitemOpen
  \bibfield  {author} {\bibinfo {author} {\bibfnamefont {R.}~\bibnamefont
  {Uola}}, \bibinfo {author} {\bibfnamefont {T.}~\bibnamefont {Bullock}},
  \bibinfo {author} {\bibfnamefont {T.}~\bibnamefont {Kraft}}, \bibinfo
  {author} {\bibfnamefont {J.-P.}\ \bibnamefont {Pellonpää}},\ and\ \bibinfo
  {author} {\bibfnamefont {N.}~\bibnamefont {Brunner}},\ }\bibfield  {title}
  {\bibinfo {title} {All quantum resources provide an advantage in exclusion
  tasks},\ }\href {https://doi.org/10.1103/physrevlett.125.110402} {\bibfield
  {journal} {\bibinfo  {journal} {Physical Review Letters}\ }\textbf {\bibinfo
  {volume} {125}},\ \bibinfo {pages} {110402} (\bibinfo {year}
  {2020}{\natexlab{a}})}\BibitemShut {NoStop}%
\bibitem [{\citenamefont {Buscemi}\ \emph {et~al.}(2020)\citenamefont
  {Buscemi}, \citenamefont {Chitambar},\ and\ \citenamefont
  {Zhou}}]{Buscemi20}%
  \BibitemOpen
  \bibfield  {author} {\bibinfo {author} {\bibfnamefont {F.}~\bibnamefont
  {Buscemi}}, \bibinfo {author} {\bibfnamefont {E.}~\bibnamefont {Chitambar}},\
  and\ \bibinfo {author} {\bibfnamefont {W.}~\bibnamefont {Zhou}},\ }\bibfield
  {title} {\bibinfo {title} {Complete resource theory of quantum
  incompatibility as quantum programmability},\ }\href
  {https://doi.org/10.1103/PhysRevLett.124.120401} {\bibfield  {journal}
  {\bibinfo  {journal} {Physical Review Letters}\ }\textbf {\bibinfo {volume}
  {124}},\ \bibinfo {pages} {120401} (\bibinfo {year} {2020})}\BibitemShut
  {NoStop}%
\bibitem [{\citenamefont {Uola}\ \emph {et~al.}(2022)\citenamefont {Uola},
  \citenamefont {Haapasalo}, \citenamefont {Pellonpää},\ and\ \citenamefont
  {Kuusela}}]{Uola22}%
  \BibitemOpen
  \bibfield  {author} {\bibinfo {author} {\bibfnamefont {R.}~\bibnamefont
  {Uola}}, \bibinfo {author} {\bibfnamefont {E.}~\bibnamefont {Haapasalo}},
  \bibinfo {author} {\bibfnamefont {J.-P.}\ \bibnamefont {Pellonpää}},\ and\
  \bibinfo {author} {\bibfnamefont {T.}~\bibnamefont {Kuusela}},\ }\href@noop
  {} {\bibinfo {title} {Retrievability of information in quantum and realistic
  hidden variable theories}} (\bibinfo {year} {2022}),\ \Eprint
  {https://arxiv.org/abs/2212.02815} {arXiv:2212.02815 [quant-ph]} \BibitemShut
  {NoStop}%
\bibitem [{\citenamefont {Midgley}\ \emph {et~al.}(2010)\citenamefont
  {Midgley}, \citenamefont {Ferris},\ and\ \citenamefont {Olsen}}]{Midgley10}%
  \BibitemOpen
  \bibfield  {author} {\bibinfo {author} {\bibfnamefont {S.~L.~W.}\
  \bibnamefont {Midgley}}, \bibinfo {author} {\bibfnamefont {A.~J.}\
  \bibnamefont {Ferris}},\ and\ \bibinfo {author} {\bibfnamefont {M.~K.}\
  \bibnamefont {Olsen}},\ }\bibfield  {title} {\bibinfo {title} {Asymmetric
  gaussian steering: When alice and bob disagree},\ }\href
  {https://doi.org/10.1103/PhysRevA.81.022101} {\bibfield  {journal} {\bibinfo
  {journal} {Phys. Rev. A}\ }\textbf {\bibinfo {volume} {81}},\ \bibinfo
  {pages} {022101} (\bibinfo {year} {2010})}\BibitemShut {NoStop}%
\bibitem [{\citenamefont {Olsen}(2013)}]{Olsen2013}%
  \BibitemOpen
  \bibfield  {author} {\bibinfo {author} {\bibfnamefont {M.~K.}\ \bibnamefont
  {Olsen}},\ }\bibfield  {title} {\bibinfo {title} {Asymmetric gaussian
  harmonic steering in second-harmonic generation},\ }\href
  {https://doi.org/10.1103/PhysRevA.88.051802} {\bibfield  {journal} {\bibinfo
  {journal} {Phys. Rev. A}\ }\textbf {\bibinfo {volume} {88}},\ \bibinfo
  {pages} {051802} (\bibinfo {year} {2013})}\BibitemShut {NoStop}%
\bibitem [{\citenamefont {Bowles}\ \emph {et~al.}(2014)\citenamefont {Bowles},
  \citenamefont {V\'ertesi}, \citenamefont {Quintino},\ and\ \citenamefont
  {Brunner}}]{Bowles2014}%
  \BibitemOpen
  \bibfield  {author} {\bibinfo {author} {\bibfnamefont {J.}~\bibnamefont
  {Bowles}}, \bibinfo {author} {\bibfnamefont {T.}~\bibnamefont {V\'ertesi}},
  \bibinfo {author} {\bibfnamefont {M.~T.}\ \bibnamefont {Quintino}},\ and\
  \bibinfo {author} {\bibfnamefont {N.}~\bibnamefont {Brunner}},\ }\bibfield
  {title} {\bibinfo {title} {One-way einstein-podolsky-rosen steering},\ }\href
  {https://doi.org/10.1103/PhysRevLett.112.200402} {\bibfield  {journal}
  {\bibinfo  {journal} {Phys. Rev. Lett.}\ }\textbf {\bibinfo {volume} {112}},\
  \bibinfo {pages} {200402} (\bibinfo {year} {2014})}\BibitemShut {NoStop}%
\bibitem [{\citenamefont {Evans}\ and\ \citenamefont
  {Wiseman}(2014)}]{Evans2014}%
  \BibitemOpen
  \bibfield  {author} {\bibinfo {author} {\bibfnamefont {D.~A.}\ \bibnamefont
  {Evans}}\ and\ \bibinfo {author} {\bibfnamefont {H.~M.}\ \bibnamefont
  {Wiseman}},\ }\bibfield  {title} {\bibinfo {title} {Optimal measurements for
  tests of einstein-podolsky-rosen steering with no detection loophole using
  two-qubit werner states},\ }\href
  {https://doi.org/10.1103/PhysRevA.90.012114} {\bibfield  {journal} {\bibinfo
  {journal} {Phys. Rev. A}\ }\textbf {\bibinfo {volume} {90}},\ \bibinfo
  {pages} {012114} (\bibinfo {year} {2014})}\BibitemShut {NoStop}%
\bibitem [{\citenamefont {Skrzypczyk}\ \emph {et~al.}(2014)\citenamefont
  {Skrzypczyk}, \citenamefont {Navascu\'es},\ and\ \citenamefont
  {Cavalcanti}}]{Skrzypczyk2014}%
  \BibitemOpen
  \bibfield  {author} {\bibinfo {author} {\bibfnamefont {P.}~\bibnamefont
  {Skrzypczyk}}, \bibinfo {author} {\bibfnamefont {M.}~\bibnamefont
  {Navascu\'es}},\ and\ \bibinfo {author} {\bibfnamefont {D.}~\bibnamefont
  {Cavalcanti}},\ }\bibfield  {title} {\bibinfo {title} {Quantifying
  einstein-podolsky-rosen steering},\ }\href
  {https://doi.org/10.1103/PhysRevLett.112.180404} {\bibfield  {journal}
  {\bibinfo  {journal} {Phys. Rev. Lett.}\ }\textbf {\bibinfo {volume} {112}},\
  \bibinfo {pages} {180404} (\bibinfo {year} {2014})}\BibitemShut {NoStop}%
\bibitem [{\citenamefont {Bowles}\ \emph {et~al.}(2016)\citenamefont {Bowles},
  \citenamefont {Hirsch}, \citenamefont {Quintino},\ and\ \citenamefont
  {Brunner}}]{Bowles2016}%
  \BibitemOpen
  \bibfield  {author} {\bibinfo {author} {\bibfnamefont {J.}~\bibnamefont
  {Bowles}}, \bibinfo {author} {\bibfnamefont {F.}~\bibnamefont {Hirsch}},
  \bibinfo {author} {\bibfnamefont {M.~T.}\ \bibnamefont {Quintino}},\ and\
  \bibinfo {author} {\bibfnamefont {N.}~\bibnamefont {Brunner}},\ }\bibfield
  {title} {\bibinfo {title} {Sufficient criterion for guaranteeing that a
  two-qubit state is unsteerable},\ }\href
  {https://doi.org/10.1103/PhysRevA.93.022121} {\bibfield  {journal} {\bibinfo
  {journal} {Phys. Rev. A}\ }\textbf {\bibinfo {volume} {93}},\ \bibinfo
  {pages} {022121} (\bibinfo {year} {2016})}\BibitemShut {NoStop}%
\bibitem [{\citenamefont {Sekatski}\ \emph {et~al.}(2023)\citenamefont
  {Sekatski}, \citenamefont {Giraud}, \citenamefont {Uola},\ and\ \citenamefont
  {Brunner}}]{Sekatski2023}%
  \BibitemOpen
  \bibfield  {author} {\bibinfo {author} {\bibfnamefont {P.}~\bibnamefont
  {Sekatski}}, \bibinfo {author} {\bibfnamefont {F.}~\bibnamefont {Giraud}},
  \bibinfo {author} {\bibfnamefont {R.}~\bibnamefont {Uola}},\ and\ \bibinfo
  {author} {\bibfnamefont {N.}~\bibnamefont {Brunner}},\ }\href@noop {}
  {\bibinfo {title} {Unlimited one-way steering}} (\bibinfo {year} {2023}),\
  \Eprint {https://arxiv.org/abs/2304.03888} {arXiv:2304.03888 [quant-ph]}
  \BibitemShut {NoStop}%
\bibitem [{\citenamefont {Branciard}\ \emph {et~al.}(2012)\citenamefont
  {Branciard}, \citenamefont {Cavalcanti}, \citenamefont {Walborn},
  \citenamefont {Scarani},\ and\ \citenamefont {Wiseman}}]{branciard2012}%
  \BibitemOpen
  \bibfield  {author} {\bibinfo {author} {\bibfnamefont {C.}~\bibnamefont
  {Branciard}}, \bibinfo {author} {\bibfnamefont {E.~G.}\ \bibnamefont
  {Cavalcanti}}, \bibinfo {author} {\bibfnamefont {S.~P.}\ \bibnamefont
  {Walborn}}, \bibinfo {author} {\bibfnamefont {V.}~\bibnamefont {Scarani}},\
  and\ \bibinfo {author} {\bibfnamefont {H.~M.}\ \bibnamefont {Wiseman}},\
  }\bibfield  {title} {\bibinfo {title} {One-sided device-independent quantum
  key distribution: Security, feasibility, and the connection with steering},\
  }\href {https://doi.org/10.1103/PhysRevA.85.010301} {\bibfield  {journal}
  {\bibinfo  {journal} {Physical Review A}\ }\textbf {\bibinfo {volume} {85}},\
  \bibinfo {pages} {010301} (\bibinfo {year} {2012})}\BibitemShut {NoStop}%
\bibitem [{\citenamefont {Law}\ \emph {et~al.}(2014)\citenamefont {Law},
  \citenamefont {Bancal}, \citenamefont {Scarani} \emph {et~al.}}]{law2014}%
  \BibitemOpen
  \bibfield  {author} {\bibinfo {author} {\bibfnamefont {Y.~Z.}\ \bibnamefont
  {Law}}, \bibinfo {author} {\bibfnamefont {J.-D.}\ \bibnamefont {Bancal}},
  \bibinfo {author} {\bibfnamefont {V.}~\bibnamefont {Scarani}}, \emph
  {et~al.},\ }\bibfield  {title} {\bibinfo {title} {Quantum randomness
  extraction for various levels of characterization of the devices},\ }\href
  {https://doi.org/10.1088/1751-8113/47/42/424028} {\bibfield  {journal}
  {\bibinfo  {journal} {Journal of Physics A: Mathematical and Theoretical}\
  }\textbf {\bibinfo {volume} {47}},\ \bibinfo {pages} {424028} (\bibinfo
  {year} {2014})}\BibitemShut {NoStop}%
\bibitem [{\citenamefont {Passaro}\ \emph {et~al.}(2015)\citenamefont
  {Passaro}, \citenamefont {Cavalcanti}, \citenamefont {Skrzypczyk},\ and\
  \citenamefont {Ac{\'{\i} }n}}]{Passaro2015}%
  \BibitemOpen
  \bibfield  {author} {\bibinfo {author} {\bibfnamefont {E.}~\bibnamefont
  {Passaro}}, \bibinfo {author} {\bibfnamefont {D.}~\bibnamefont {Cavalcanti}},
  \bibinfo {author} {\bibfnamefont {P.}~\bibnamefont {Skrzypczyk}},\ and\
  \bibinfo {author} {\bibfnamefont {A.}~\bibnamefont {Ac{\'{\i} }n}},\
  }\bibfield  {title} {\bibinfo {title} {Optimal randomness certification in
  the quantum steering and prepare-and-measure scenarios},\ }\href
  {https://doi.org/10.1088/1367-2630/17/11/113010} {\bibfield  {journal}
  {\bibinfo  {journal} {New Journal of Physics}\ }\textbf {\bibinfo {volume}
  {17}},\ \bibinfo {pages} {113010} (\bibinfo {year} {2015})}\BibitemShut
  {NoStop}%
\bibitem [{\citenamefont {Skrzypczyk}\ and\ \citenamefont
  {Cavalcanti}(2018)}]{Skrzypczyk2018}%
  \BibitemOpen
  \bibfield  {author} {\bibinfo {author} {\bibfnamefont {P.}~\bibnamefont
  {Skrzypczyk}}\ and\ \bibinfo {author} {\bibfnamefont {D.}~\bibnamefont
  {Cavalcanti}},\ }\bibfield  {title} {\bibinfo {title} {{Maximal Randomness
  Generation from Steering Inequality Violations Using Qudits}},\ }\href
  {https://doi.org/10.1103/physrevlett.120.260401} {\bibfield  {journal}
  {\bibinfo  {journal} {Physical Review Letters}\ }\textbf {\bibinfo {volume}
  {120}},\ \bibinfo {pages} {260401} (\bibinfo {year} {2018})}\BibitemShut
  {NoStop}%
\bibitem [{\citenamefont {Xiang}\ \emph {et~al.}(2017)\citenamefont {Xiang},
  \citenamefont {Kogias}, \citenamefont {Adesso},\ and\ \citenamefont
  {He}}]{xiang2017}%
  \BibitemOpen
  \bibfield  {author} {\bibinfo {author} {\bibfnamefont {Y.}~\bibnamefont
  {Xiang}}, \bibinfo {author} {\bibfnamefont {I.}~\bibnamefont {Kogias}},
  \bibinfo {author} {\bibfnamefont {G.}~\bibnamefont {Adesso}},\ and\ \bibinfo
  {author} {\bibfnamefont {Q.}~\bibnamefont {He}},\ }\bibfield  {title}
  {\bibinfo {title} {Multipartite gaussian steering: Monogamy constraints and
  quantum cryptography applications},\ }\href
  {https://doi.org/10.1103/PhysRevA.95.010101} {\bibfield  {journal} {\bibinfo
  {journal} {Physical Review A}\ }\textbf {\bibinfo {volume} {95}},\ \bibinfo
  {pages} {010101} (\bibinfo {year} {2017})}\BibitemShut {NoStop}%
\bibitem [{\citenamefont {Kogias}\ \emph {et~al.}(2017)\citenamefont {Kogias},
  \citenamefont {Xiang}, \citenamefont {He},\ and\ \citenamefont
  {Adesso}}]{Kogias17}%
  \BibitemOpen
  \bibfield  {author} {\bibinfo {author} {\bibfnamefont {I.}~\bibnamefont
  {Kogias}}, \bibinfo {author} {\bibfnamefont {Y.}~\bibnamefont {Xiang}},
  \bibinfo {author} {\bibfnamefont {Q.}~\bibnamefont {He}},\ and\ \bibinfo
  {author} {\bibfnamefont {G.}~\bibnamefont {Adesso}},\ }\bibfield  {title}
  {\bibinfo {title} {Unconditional security of entanglement-based
  continuous-variable quantum secret sharing},\ }\href
  {https://doi.org/10.1103/PhysRevA.95.012315} {\bibfield  {journal} {\bibinfo
  {journal} {Physical Review A}\ }\textbf {\bibinfo {volume} {95}},\ \bibinfo
  {pages} {012315} (\bibinfo {year} {2017})}\BibitemShut {NoStop}%
\bibitem [{\citenamefont {Cavalcanti}\ and\ \citenamefont
  {Skrzypczyk}(2016{\natexlab{a}})}]{Cavalcanti2016}%
  \BibitemOpen
  \bibfield  {author} {\bibinfo {author} {\bibfnamefont {D.}~\bibnamefont
  {Cavalcanti}}\ and\ \bibinfo {author} {\bibfnamefont {P.}~\bibnamefont
  {Skrzypczyk}},\ }\bibfield  {title} {\bibinfo {title} {Quantum steering: a
  review with focus on semidefinite programming},\ }\href
  {https://doi.org/10.1088/1361-6633/80/2/024001} {\bibfield  {journal}
  {\bibinfo  {journal} {Reports on Progress in Physics}\ }\textbf {\bibinfo
  {volume} {80}},\ \bibinfo {pages} {024001} (\bibinfo {year}
  {2016}{\natexlab{a}})}\BibitemShut {NoStop}%
\bibitem [{\citenamefont {Uola}\ \emph
  {et~al.}(2020{\natexlab{b}})\citenamefont {Uola}, \citenamefont {Costa},
  \citenamefont {Nguyen},\ and\ \citenamefont {G{\"u}hne}}]{steering-review}%
  \BibitemOpen
  \bibfield  {author} {\bibinfo {author} {\bibfnamefont {R.}~\bibnamefont
  {Uola}}, \bibinfo {author} {\bibfnamefont {A.~C.}\ \bibnamefont {Costa}},
  \bibinfo {author} {\bibfnamefont {H.~C.}\ \bibnamefont {Nguyen}},\ and\
  \bibinfo {author} {\bibfnamefont {O.}~\bibnamefont {G{\"u}hne}},\ }\bibfield
  {title} {\bibinfo {title} {Quantum steering},\ }\href
  {https://doi.org/10.1103/RevModPhys.92.015001} {\bibfield  {journal}
  {\bibinfo  {journal} {Reviews of Modern Physics}\ }\textbf {\bibinfo {volume}
  {92}},\ \bibinfo {pages} {015001} (\bibinfo {year}
  {2020}{\natexlab{b}})}\BibitemShut {NoStop}%
\bibitem [{\citenamefont {Chen}\ \emph {et~al.}(2016)\citenamefont {Chen},
  \citenamefont {Budroni}, \citenamefont {Liang},\ and\ \citenamefont
  {Chen}}]{Chen2016}%
  \BibitemOpen
  \bibfield  {author} {\bibinfo {author} {\bibfnamefont {S.-L.}\ \bibnamefont
  {Chen}}, \bibinfo {author} {\bibfnamefont {C.}~\bibnamefont {Budroni}},
  \bibinfo {author} {\bibfnamefont {Y.-C.}\ \bibnamefont {Liang}},\ and\
  \bibinfo {author} {\bibfnamefont {Y.-N.}\ \bibnamefont {Chen}},\ }\bibfield
  {title} {\bibinfo {title} {Natural framework for device-independent
  quantification of quantum steerability, measurement incompatibility, and
  self-testing},\ }\href {https://doi.org/10.1103/physrevlett.116.240401}
  {\bibfield  {journal} {\bibinfo  {journal} {Physical Review Letters}\
  }\textbf {\bibinfo {volume} {116}},\ \bibinfo {pages} {240401} (\bibinfo
  {year} {2016})}\BibitemShut {NoStop}%
\bibitem [{\citenamefont {Cavalcanti}\ and\ \citenamefont
  {Skrzypczyk}(2016{\natexlab{b}})}]{Cavalcanti16quantitative}%
  \BibitemOpen
  \bibfield  {author} {\bibinfo {author} {\bibfnamefont {D.}~\bibnamefont
  {Cavalcanti}}\ and\ \bibinfo {author} {\bibfnamefont {P.}~\bibnamefont
  {Skrzypczyk}},\ }\bibfield  {title} {\bibinfo {title} {Quantitative relations
  between measurement incompatibility, quantum steering, and nonlocality},\
  }\href {https://doi.org/10.1103/PhysRevA.93.052112} {\bibfield  {journal}
  {\bibinfo  {journal} {Physical Review A}\ }\textbf {\bibinfo {volume} {93}},\
  \bibinfo {pages} {052112} (\bibinfo {year} {2016}{\natexlab{b}})}\BibitemShut
  {NoStop}%
\bibitem [{\citenamefont {Heinosaari}\ \emph {et~al.}(2008)\citenamefont
  {Heinosaari}, \citenamefont {Reitzner},\ and\ \citenamefont
  {Stano}}]{Heinosaari2008}%
  \BibitemOpen
  \bibfield  {author} {\bibinfo {author} {\bibfnamefont {T.}~\bibnamefont
  {Heinosaari}}, \bibinfo {author} {\bibfnamefont {D.}~\bibnamefont
  {Reitzner}},\ and\ \bibinfo {author} {\bibfnamefont {P.}~\bibnamefont
  {Stano}},\ }\bibfield  {title} {\bibinfo {title} {Notes on joint
  measurability of quantum observables},\ }\href
  {https://doi.org/10.1007/s10701-008-9256-7} {\bibfield  {journal} {\bibinfo
  {journal} {Foundations of Physics}\ }\textbf {\bibinfo {volume} {38}},\
  \bibinfo {pages} {1133} (\bibinfo {year} {2008})}\BibitemShut {NoStop}%
\bibitem [{\citenamefont {Uola}\ \emph {et~al.}(2015)\citenamefont {Uola},
  \citenamefont {Budroni}, \citenamefont {G{\"u}hne},\ and\ \citenamefont
  {Pellonp{\"a}{\"a}}}]{uola2015one}%
  \BibitemOpen
  \bibfield  {author} {\bibinfo {author} {\bibfnamefont {R.}~\bibnamefont
  {Uola}}, \bibinfo {author} {\bibfnamefont {C.}~\bibnamefont {Budroni}},
  \bibinfo {author} {\bibfnamefont {O.}~\bibnamefont {G{\"u}hne}},\ and\
  \bibinfo {author} {\bibfnamefont {J.-P.}\ \bibnamefont {Pellonp{\"a}{\"a}}},\
  }\bibfield  {title} {\bibinfo {title} {One-to-one mapping between steering
  and joint measurability problems},\ }\href
  {https://doi.org/10.1103/PhysRevLett.115.230402} {\bibfield  {journal}
  {\bibinfo  {journal} {Physical Review Letters}\ }\textbf {\bibinfo {volume}
  {115}},\ \bibinfo {pages} {230402} (\bibinfo {year} {2015})}\BibitemShut
  {NoStop}%
\bibitem [{\citenamefont {Filippov}\ \emph {et~al.}(2018)\citenamefont
  {Filippov}, \citenamefont {Heinosaari},\ and\ \citenamefont
  {Lepp\"aj\"arvi}}]{Filippov2018}%
  \BibitemOpen
  \bibfield  {author} {\bibinfo {author} {\bibfnamefont {S.~N.}\ \bibnamefont
  {Filippov}}, \bibinfo {author} {\bibfnamefont {T.}~\bibnamefont
  {Heinosaari}},\ and\ \bibinfo {author} {\bibfnamefont {L.}~\bibnamefont
  {Lepp\"aj\"arvi}},\ }\bibfield  {title} {\bibinfo {title} {Simulability of
  observables in general probabilistic theories},\ }\href
  {https://doi.org/10.1103/PhysRevA.97.062102} {\bibfield  {journal} {\bibinfo
  {journal} {Physical Review A}\ }\textbf {\bibinfo {volume} {97}},\ \bibinfo
  {pages} {062102} (\bibinfo {year} {2018})}\BibitemShut {NoStop}%
\bibitem [{\citenamefont {Doherty}\ \emph {et~al.}(2004)\citenamefont
  {Doherty}, \citenamefont {Parrilo},\ and\ \citenamefont {Spedalieri}}]{dps}%
  \BibitemOpen
  \bibfield  {author} {\bibinfo {author} {\bibfnamefont {A.~C.}\ \bibnamefont
  {Doherty}}, \bibinfo {author} {\bibfnamefont {P.~A.}\ \bibnamefont
  {Parrilo}},\ and\ \bibinfo {author} {\bibfnamefont {F.~M.}\ \bibnamefont
  {Spedalieri}},\ }\bibfield  {title} {\bibinfo {title} {Complete family of
  separability criteria},\ }\href {https://doi.org/10.1103/PhysRevA.69.022308}
  {\bibfield  {journal} {\bibinfo  {journal} {Physical Review A}\ }\textbf
  {\bibinfo {volume} {69}},\ \bibinfo {pages} {022308} (\bibinfo {year}
  {2004})}\BibitemShut {NoStop}%
\bibitem [{\citenamefont {Aubrun}\ \emph {et~al.}(2022)\citenamefont {Aubrun},
  \citenamefont {M{\"u}ller-Hermes},\ and\ \citenamefont
  {Pl{\'a}vala}}]{aubrun2022monogamy}%
  \BibitemOpen
  \bibfield  {author} {\bibinfo {author} {\bibfnamefont {G.}~\bibnamefont
  {Aubrun}}, \bibinfo {author} {\bibfnamefont {A.}~\bibnamefont
  {M{\"u}ller-Hermes}},\ and\ \bibinfo {author} {\bibfnamefont
  {M.}~\bibnamefont {Pl{\'a}vala}},\ }\href@noop {} {\bibinfo {title} {Monogamy
  of entanglement between cones}} (\bibinfo {year} {2022}),\ \Eprint
  {https://arxiv.org/abs/2206.11805} {2206.11805} \BibitemShut {NoStop}%
\bibitem [{\citenamefont {Nadlinger}\ \emph {et~al.}(2022)\citenamefont
  {Nadlinger}, \citenamefont {Drmota}, \citenamefont {Nichol}, \citenamefont
  {Araneda}, \citenamefont {Main}, \citenamefont {Srinivas}, \citenamefont
  {Lucas}, \citenamefont {Ballance}, \citenamefont {Ivanov}, \citenamefont
  {Tan} \emph {et~al.}}]{nadlinger2022experimental}%
  \BibitemOpen
  \bibfield  {author} {\bibinfo {author} {\bibfnamefont {D.~P.}\ \bibnamefont
  {Nadlinger}}, \bibinfo {author} {\bibfnamefont {P.}~\bibnamefont {Drmota}},
  \bibinfo {author} {\bibfnamefont {B.~C.}\ \bibnamefont {Nichol}}, \bibinfo
  {author} {\bibfnamefont {G.}~\bibnamefont {Araneda}}, \bibinfo {author}
  {\bibfnamefont {D.}~\bibnamefont {Main}}, \bibinfo {author} {\bibfnamefont
  {R.}~\bibnamefont {Srinivas}}, \bibinfo {author} {\bibfnamefont {D.~M.}\
  \bibnamefont {Lucas}}, \bibinfo {author} {\bibfnamefont {C.~J.}\ \bibnamefont
  {Ballance}}, \bibinfo {author} {\bibfnamefont {K.}~\bibnamefont {Ivanov}},
  \bibinfo {author} {\bibfnamefont {E.-Z.}\ \bibnamefont {Tan}}, \emph
  {et~al.},\ }\bibfield  {title} {\bibinfo {title} {Experimental quantum key
  distribution certified by bell's theorem},\ }\href
  {https://doi.org/10.1038/s41586-022-04941-5} {\bibfield  {journal} {\bibinfo
  {journal} {Nature}\ }\textbf {\bibinfo {volume} {607}},\ \bibinfo {pages}
  {682} (\bibinfo {year} {2022})}\BibitemShut {NoStop}%
\bibitem [{\citenamefont {Zhang}\ \emph {et~al.}(2022)\citenamefont {Zhang},
  \citenamefont {van Leent}, \citenamefont {Redeker}, \citenamefont {Garthoff},
  \citenamefont {Schwonnek}, \citenamefont {Fertig}, \citenamefont {Eppelt},
  \citenamefont {Rosenfeld}, \citenamefont {Scarani}, \citenamefont {Lim} \emph
  {et~al.}}]{zhang2022device}%
  \BibitemOpen
  \bibfield  {author} {\bibinfo {author} {\bibfnamefont {W.}~\bibnamefont
  {Zhang}}, \bibinfo {author} {\bibfnamefont {T.}~\bibnamefont {van Leent}},
  \bibinfo {author} {\bibfnamefont {K.}~\bibnamefont {Redeker}}, \bibinfo
  {author} {\bibfnamefont {R.}~\bibnamefont {Garthoff}}, \bibinfo {author}
  {\bibfnamefont {R.}~\bibnamefont {Schwonnek}}, \bibinfo {author}
  {\bibfnamefont {F.}~\bibnamefont {Fertig}}, \bibinfo {author} {\bibfnamefont
  {S.}~\bibnamefont {Eppelt}}, \bibinfo {author} {\bibfnamefont
  {W.}~\bibnamefont {Rosenfeld}}, \bibinfo {author} {\bibfnamefont
  {V.}~\bibnamefont {Scarani}}, \bibinfo {author} {\bibfnamefont {C.~C.-W.}\
  \bibnamefont {Lim}}, \emph {et~al.},\ }\bibfield  {title} {\bibinfo {title}
  {A device-independent quantum key distribution system for distant users},\
  }\href {https://doi.org/10.1038/s41586-022-04891-y} {\bibfield  {journal}
  {\bibinfo  {journal} {Nature}\ }\textbf {\bibinfo {volume} {607}},\ \bibinfo
  {pages} {687} (\bibinfo {year} {2022})}\BibitemShut {NoStop}%
\bibitem [{\citenamefont {Zapatero}\ \emph {et~al.}(2023)\citenamefont
  {Zapatero}, \citenamefont {van Leent}, \citenamefont {Arnon-Friedman},
  \citenamefont {Liu}, \citenamefont {Zhang}, \citenamefont {Weinfurter},\ and\
  \citenamefont {Curty}}]{zapatero2023advances}%
  \BibitemOpen
  \bibfield  {author} {\bibinfo {author} {\bibfnamefont {V.}~\bibnamefont
  {Zapatero}}, \bibinfo {author} {\bibfnamefont {T.}~\bibnamefont {van Leent}},
  \bibinfo {author} {\bibfnamefont {R.}~\bibnamefont {Arnon-Friedman}},
  \bibinfo {author} {\bibfnamefont {W.-Z.}\ \bibnamefont {Liu}}, \bibinfo
  {author} {\bibfnamefont {Q.}~\bibnamefont {Zhang}}, \bibinfo {author}
  {\bibfnamefont {H.}~\bibnamefont {Weinfurter}},\ and\ \bibinfo {author}
  {\bibfnamefont {M.}~\bibnamefont {Curty}},\ }\bibfield  {title} {\bibinfo
  {title} {Advances in device-independent quantum key distribution},\ }\href
  {https://doi.org/10.1038/s41534-023-00684-x} {\bibfield  {journal} {\bibinfo
  {journal} {npj Quantum Information}\ }\textbf {\bibinfo {volume} {9}},\
  \bibinfo {pages} {10} (\bibinfo {year} {2023})}\BibitemShut {NoStop}%
\bibitem [{\citenamefont {Miklin}\ \emph {et~al.}(2022)\citenamefont {Miklin},
  \citenamefont {Chaturvedi}, \citenamefont {Bourennane}, \citenamefont
  {Paw{\l}owski},\ and\ \citenamefont {Cabello}}]{miklin2022exponentially}%
  \BibitemOpen
  \bibfield  {author} {\bibinfo {author} {\bibfnamefont {N.}~\bibnamefont
  {Miklin}}, \bibinfo {author} {\bibfnamefont {A.}~\bibnamefont {Chaturvedi}},
  \bibinfo {author} {\bibfnamefont {M.}~\bibnamefont {Bourennane}}, \bibinfo
  {author} {\bibfnamefont {M.}~\bibnamefont {Paw{\l}owski}},\ and\ \bibinfo
  {author} {\bibfnamefont {A.}~\bibnamefont {Cabello}},\ }\bibfield  {title}
  {\bibinfo {title} {Exponentially decreasing critical detection efficiency for
  any bell inequality},\ }\href
  {https://doi.org/10.1103/PhysRevLett.129.230403} {\bibfield  {journal}
  {\bibinfo  {journal} {Physical Review Letters}\ }\textbf {\bibinfo {volume}
  {129}},\ \bibinfo {pages} {230403} (\bibinfo {year} {2022})}\BibitemShut
  {NoStop}%
\bibitem [{\citenamefont {Xu}\ \emph {et~al.}(2023)\citenamefont {Xu},
  \citenamefont {Steinberg}, \citenamefont {Singh}, \citenamefont
  {L{\'o}pez-Tarrida}, \citenamefont {Portillo},\ and\ \citenamefont
  {Cabello}}]{xu2023graph}%
  \BibitemOpen
  \bibfield  {author} {\bibinfo {author} {\bibfnamefont {Z.-P.}\ \bibnamefont
  {Xu}}, \bibinfo {author} {\bibfnamefont {J.}~\bibnamefont {Steinberg}},
  \bibinfo {author} {\bibfnamefont {J.}~\bibnamefont {Singh}}, \bibinfo
  {author} {\bibfnamefont {A.~J.}\ \bibnamefont {L{\'o}pez-Tarrida}}, \bibinfo
  {author} {\bibfnamefont {J.~R.}\ \bibnamefont {Portillo}},\ and\ \bibinfo
  {author} {\bibfnamefont {A.}~\bibnamefont {Cabello}},\ }\bibfield  {title}
  {\bibinfo {title} {Graph-theoretic approach to bell experiments with low
  detection efficiency},\ }\href {https://doi.org/10.22331/q-2023-02-16-922}
  {\bibfield  {journal} {\bibinfo  {journal} {Quantum}\ }\textbf {\bibinfo
  {volume} {7}},\ \bibinfo {pages} {922} (\bibinfo {year} {2023})}\BibitemShut
  {NoStop}%
\bibitem [{\citenamefont {Brown}\ \emph {et~al.}(2021)\citenamefont {Brown},
  \citenamefont {Fawzi},\ and\ \citenamefont {Fawzi}}]{brown2021computing}%
  \BibitemOpen
  \bibfield  {author} {\bibinfo {author} {\bibfnamefont {P.}~\bibnamefont
  {Brown}}, \bibinfo {author} {\bibfnamefont {H.}~\bibnamefont {Fawzi}},\ and\
  \bibinfo {author} {\bibfnamefont {O.}~\bibnamefont {Fawzi}},\ }\bibfield
  {title} {\bibinfo {title} {Computing conditional entropies for quantum
  correlations},\ }\href {https://doi.org/10.1038/s41467-020-20018-1}
  {\bibfield  {journal} {\bibinfo  {journal} {Nature communications}\ }\textbf
  {\bibinfo {volume} {12}},\ \bibinfo {pages} {575} (\bibinfo {year}
  {2021})}\BibitemShut {NoStop}%
\bibitem [{\citenamefont {Gonzales-Ureta}\ \emph {et~al.}(2021)\citenamefont
  {Gonzales-Ureta}, \citenamefont {Predojevi{\'c}},\ and\ \citenamefont
  {Cabello}}]{gonzales2021device}%
  \BibitemOpen
  \bibfield  {author} {\bibinfo {author} {\bibfnamefont {J.~R.}\ \bibnamefont
  {Gonzales-Ureta}}, \bibinfo {author} {\bibfnamefont {A.}~\bibnamefont
  {Predojevi{\'c}}},\ and\ \bibinfo {author} {\bibfnamefont {A.}~\bibnamefont
  {Cabello}},\ }\bibfield  {title} {\bibinfo {title} {Device-independent
  quantum key distribution based on bell inequalities with more than two inputs
  and two outputs},\ }\href {https://doi.org/10.1103/PhysRevA.103.052436}
  {\bibfield  {journal} {\bibinfo  {journal} {Physical Review A}\ }\textbf
  {\bibinfo {volume} {103}},\ \bibinfo {pages} {052436} (\bibinfo {year}
  {2021})}\BibitemShut {NoStop}%
\end{thebibliography}%

\onecolumngrid
\appendix
\section{Sufficient conditions for $2$-simulability based on $2$-compatibility}

We make the following simple observation: Consider a set of $k$-compatible measurements $\{M_{a|x}\}_{a,x}$ with a $k$-copy joint observable $G_\lambda$ so that  
\begin{equation}\label{eq:k-comp}
    \tr[\rho M_{a|x}] = \sum_{\lambda} p(a|x,\lambda) \tr[\rho^{\otimes k} G_{\lambda}]
\end{equation}
for all states $\rho$ for some conditional probability distributions $p$. Suppose now that the POVM $\{G_{\lambda}\}_{\lambda}$ is a mixture of $k$ POVMs $\{\id^{\otimes y-1} \otimes B_{\lambda|y} \otimes \id^{\otimes k-y}\}_{\lambda,y}$, where $y \in \{1, \ldots,k\}$ and $\{B_{\lambda|y}\}_{b,y}$ are some $k$ POVMs. Thus, there exists a probability distribution $q$ such that
\begin{equation}
G_\lambda = \sum_{y=1}^k q(y) \left( \id^{\otimes y-1} \otimes B_{\lambda|y} \otimes \id^{\otimes k-y} \right).
\end{equation} 
Now we see that the original measurements $\{M_{a|x}\}_{a,x}$ are in fact $k$-simulable:
\begin{equation}
 \tr[\rho M_{a|x}] = \sum_{\lambda} p(a|x,\lambda) \tr[\rho^{\otimes k} G_{\lambda}] = \tr\left[\rho \left(\sum_{y=1}^k q(y) \sum_\lambda p(a|x,\lambda)  B_{\lambda|y}\right) \right].
\end{equation}

Consider now a set of $2$-compatible measurements $\{M_{a|x}\}_{a,x}$ with a joint observable $G_\lambda$. Let us take any two different states $\rho_1$ and $\rho_2$ and consider their mixture $\rho = \mu \rho_1 +(1-\mu) \rho_2$ by some weight $\mu \in (0,1)$. By the linearity of the trace we naturally have that $\tr[\rho M_{a|x}] = \mu \tr[\rho_1 M_{a|x}] + (1- \mu) \tr[\rho_2 M_{a|x}]$ and by using Eq. \eqref{eq:k-comp} for $k=2$ one can easily rephrase the condition as
\begin{equation}\label{eq:2-comp-vanish}
	\sum_{\lambda} p(a|x,\lambda) \tr\left[ \left( \rho_1 - \rho_2 \right)^{\otimes 2}G_\lambda \right] = 0 .
\end{equation}
Since this must hold for any two states $\rho_1$ and $\rho_2$, we must have that for all $a,x, \lambda$ such that $p(a|x,\lambda)\neq0$ we have that $G_\lambda = \tilde{B}_{\lambda|1} \otimes \id_B + \id_A \otimes \tilde{B}_{\lambda|2} + W_{\lambda}$ for some operators $\tilde{B}_{\lambda|1},\tilde{B}_{\lambda|2}, W_{\lambda}$ such that in particular $\tr\left[ \left( \rho_1 - \rho_2 \right)^{\otimes 2} W_{\lambda} \right] = 0$. In the case when $W_\lambda = 0$ and the operators $\tilde{B}_{\lambda|1}$ and $\tilde{B}_{\lambda|2}$ are selfadjoint we can show that the observation described above applies.

\begin{proposition}
Let $\{M_{a|x}\}_{a,x}$ be a set of $2$-compatible POVMs with a $2$-copy joint POMV $G$. If $G_\lambda \in \linspan(\{ B_1 \otimes \id_B + \id_A \otimes B_2 : B_1 \in \mathcal{L}(\mathcal{H}_A), B_2 \in \mathcal{L}(\mathcal{H}_B), B_1, B_2 \text{ are selfadjoint}\})$ for all $\lambda$, then $\{M_{a|x}\}_{a,x}$ are $2$-simulable.
\end{proposition}
\begin{proof}
    Since $G$ is a 2-copy joint POVM for $\{M_{a|x}\}_{a,x}$, there exists a conditional probability ditribution $p$ such that $\tr[\rho M_{a|x}] = \sum_{\lambda} p(a|x,\lambda) \tr[\rho^{\otimes 2} G_{\lambda}]$ for all $a,x$. If now $G \in \linspan(\{ B_1 \otimes \id_B + \id_A \otimes B_2 : B_1 \in \mathcal{L}(\mathcal{H}_A), B_2 \in \mathcal{L}(\mathcal{H}_B), B_1, B_2 \text{ are selfadjoint}\})$, then there exists some  selfadjoint $\tilde{A}_\lambda \in \mathcal{L}(\mathcal{H}_A),  \tilde{B}_\lambda \in \mathcal{L}(\mathcal{H}_B)$ such that $G_\lambda = \tilde{A}_\lambda \otimes \id_B +  \id_A \otimes \tilde{B}_\lambda$. We see that in fact we can choose another decomposition $G_\lambda = A_\lambda \otimes \id_B +  \id_A \otimes B_\lambda$ such that $A_\lambda$ and $B_\lambda$ are positive semidefinite: Namely, let us consider the spectral decompositions $\tilde{A}_\lambda = \sum_i \nu_\lambda^{(i)} \ket{\phi_\lambda^{(i)}}\bra{\phi_\lambda^{(i)}}$ and $\tilde{B}_\lambda = \sum_j \mu_\lambda^{(j)} \ket{\psi_\lambda^{(j)}}\bra{\psi_\lambda^{(j)}}$ for some real numbers $\nu_\lambda^{(i)}, \mu_\lambda^{(j)}$ and some orthonormal bases $\{\phi_\lambda^{(i)}\}_i$ and $\{\psi_\lambda^{(j)}\}_j$ of $\mathcal{H}_A$ and $\mathcal{H}_B$, respectively. Then we clearly have that 
    \begin{equation}
        G_\lambda = \tilde{A}_\lambda \otimes \id_B +  \id_A \otimes \tilde{B}_\lambda = \sum_{i,j} (\nu_\lambda^{(i)}+ \mu_\lambda^{(j)}) \ket{\phi_\lambda^{(i)}}\ket{\psi_\lambda^{(j)}}\bra{\psi_\lambda^{(j)}}\bra{\phi_\lambda^{(i)}}.
    \end{equation}
    The above equation defines a spectral decomposition for $G_\lambda$ so that from the positive semidefiniteness of $G_\lambda$ it follows that $\nu_\lambda^{(i)}+ \mu_\lambda^{(j)} \geq 0$ for all $i,j,\lambda$. Thus, in particular we must have that either $\tilde{A}_\lambda$ or $\tilde{B}_\lambda$ or both of them are positive semidefinite for all $\lambda$. Without loss of generality, we assume that $\tilde{A}_\lambda \geq 0$ so that $\nu^{(i)}_\lambda \geq 0$ for all $i,\lambda$. By denoting $\nu^{min}_\lambda := \min_i \nu^{(i)}_\lambda$ and defining $A_\lambda := \tilde{A}_\lambda - \nu^{min}_\lambda \id_A$ we see that also $A_\lambda \geq 0$. Now we see that 
    \begin{equation*}
        G_\lambda = \tilde{A}_\lambda \otimes \id_B +  \id_A \otimes \tilde{B}_\lambda = A_\lambda \otimes \id_B +  \nu^{min}_\lambda \id_A \otimes \id_B + \id_A \otimes \tilde{B}_\lambda = A_\lambda \otimes \id_B +  \id_A \otimes \left(\nu^{min}_\lambda \id_B + \tilde{B}_\lambda \right)
    \end{equation*}
    so that by denoting $B_\lambda := \nu^{min}_\lambda \id_B + \tilde{B}_\lambda$ we must have that $B_\lambda$ is positive semidefinite because $\nu_\lambda^{min}+ \mu_\lambda^{(j)} \geq 0$ for all $j$.

    By taking the partial traces separately with respect to $\mathcal{H}_A$ and $\mathcal{H}_B$ from the normalization condition $\sum_\lambda G_\lambda = \id_A \otimes \id_B$ we see now that
    \begin{align}
        \sum_\lambda A_\lambda &=  \left(1-\frac{\sum_\lambda \tr[B_\lambda]}{d_B} \right) \id_A \\
        \sum_\lambda B_\lambda &=  \left(1-\frac{\sum_\lambda \tr[A_\lambda]}{d_A} \right) \id_B \, .
    \end{align}
    Since $A_\lambda$ and $B_\lambda$ are positive semidefinite, we must have that $q :=  \left(1-\frac{\sum_\lambda \tr[B_\lambda]}{d_B} \right) \geq 0$ and $q' := \left(1-\frac{\sum_\lambda \tr[A_\lambda]}{d_A} \right) \geq 0$, and from the normalization of $G$ it also follows that $q' = 1- q$. By denoting $C_\lambda := \frac{1}{q} A_\lambda$ and $D_\lambda := \frac{1}{1-q} B_\lambda$ whenever $q,1-q \neq 0$ and $C_\lambda := 0 =: D_\lambda$ otherwise, we  have that $\{C_\lambda\}_{\lambda}$ and $\{D_\lambda\}_{\lambda}$ are in fact two POVMs such that
    \begin{equation}
        G_\lambda = q C_\lambda \otimes \id_B + (1-q) \id_A \otimes D_\lambda
    \end{equation}
    for all $\lambda$. Now we see that we are in the case described by the observation in the beginning of this section so that $\{M_{a|x}\}_{a,x}$ are $2$-simulable.

\end{proof}

We note that both the $2$-compatibility and the previously described condition for the $2$-copy joint POVM can be checked by using SDPs.

\end{document}